\documentclass[11pt]{article}

\usepackage{amsmath,amssymb,amsfonts,amsthm,epsfig,verbatim, color}
\usepackage{times}

\def\short{0}

\newif\ifhyper\IfFileExists{hyperref.sty}{\hypertrue}{\hyperfalse}
\hypertrue
\ifhyper\usepackage{hyperref}\fi

\newtheorem{theorem}{Theorem}

\newtheorem{lemma}[theorem]{Lemma}
\newtheorem{proposition}[theorem]{Proposition}
\newtheorem{corollary}[theorem]{Corollary}
\newtheorem{claim}[theorem]{Claim}
\newtheorem{fact}[theorem]{Fact}

\newtheorem{definition}[theorem]{Definition}
\newtheorem{remark}[theorem]{Remark}

\newcommand{\good}{{{\mathrm{good}}}}
\newcommand{\w}{{{\bf W}}}
\newcommand{\kk}{{{\bf K}}}
\newcommand{\ltf}{{{\bf LTF}}}

\newcommand{\apx}[1]{\stackrel{#1}{\approx}}

\newcommand{\eps}{\epsilon}

\newcommand{\Var}{\operatorname{{\bf Var}}}

\newcommand{\sgn}{\mathrm{sign}}
\newcommand{\sign}{\mathrm{sign}}

\newcommand{\ignore}[1]{}

\newcommand{\cref}[1]{Corollary~\ref{cor:#1}}

\newcommand{\supip} {{^{i+}}}

\newcommand{\supim} {{^{i-}}}

\newcommand{\bits}{\{-1,1\}}
\newcommand{\bn}{\bits^n}

\newcommand{\R}{{\mathbb{R}}}
\newcommand{\N}{{\mathbb N}}

\newcommand{\E}{\operatorname{{\bf E}}}

\newcommand{\littlesum}{\mathop{\textstyle \sum}}

\newcommand{\poly}{\mathrm{poly}}

\newcommand{\Inf}{\mathbf{Inf}}

\newcommand{\wh}{\widehat}
\newcommand{\eqdef}{\stackrel{\textrm{def}}{=}}
\newcommand{\la}{\langle}
\newcommand{\ra}{\rangle}

\newcommand{\margarine}{\mathrm{marg}}
\newcommand{\T}{{{\bf T}}}

\renewcommand{\Pr}{\operatorname{{\bf Pr}}}
\newcommand{\dist}{\mathrm{dist}}

\newcommand{\Nd}{{\cal N}}

\newcommand{\new}[1]{{\color{red} {#1}}}
\newcommand{\newa}[1]{{{#1}}}
\newcommand{\newb}[1]{{{#1}}}

\newcommand{\inote}[1]{\footnote{{\bf [[Ilias: {#1}\bf ]] }}}

\begin{document}
\pagestyle{empty}

\title{A robust Khintchine inequality, and\\
algorithms for computing optimal constants in\\ Fourier analysis and
high-dimensional geometry }

\author{Anindya De\thanks{Research supported by NSF award CCF-0915929 and NSF  award CCF-1017403.}\\
University of California, Berkeley\\
{\tt anindya@cs.berkeley.edu}\\
\and
Ilias Diakonikolas\thanks{Research performed in part while supported by a
Simons Postdoctoral Fellowship at UC Berkeley.}\\
University of Edinburgh, Edinburgh, UK\\
{\tt ilias.d@ed.ac.uk}\\
\and Rocco A.\
Servedio\thanks{ Supported by NSF grants CCF-0915929 and
CCF-1115703.}\\
Columbia University\\
{\tt rocco@cs.columbia.edu}}

\maketitle

\setcounter{page}{0}

\thispagestyle{empty}

\begin{abstract}
This paper makes two contributions towards determining some well-studied
optimal constants in Fourier analysis \newa{of Boolean functions}
and high-dimensional geometry.

\begin{enumerate}

\item It has been known since 1994 \cite{GL:94}
that every linear threshold function has squared Fourier mass
at least $1/2$ on its degree-$0$ and degree-$1$ coefficients.
Denote the minimum such Fourier mass by
$\w^{\leq 1}[\ltf]$, where the minimum is taken over all $n$-variable
linear threshold functions and all $n \ge 0$.
Benjamini, Kalai and Schramm \cite{BKS:99} 
have conjectured that the true value of $\w^{\leq 1}[\ltf]$ is $2/\pi$.
We make progress on this conjecture by proving that $\w^{\leq 1}[\ltf]
\geq 1/2 + c$ for some absolute constant $c>0$.
The key ingredient in our proof is a ``robust'' version of the well-known
Khintchine inequality in functional analysis, which we
believe may be of independent interest.

\item
We give an algorithm with the following property:  given any $\eta > 0$,
the algorithm runs in time $2^{\poly(1/\eta)}$ and determines the value of
$\w^{\leq 1}[\ltf]$ up to an additive error of $\pm\eta$.  We give a similar
\newa{$2^{{\poly(1/\eta)}}$}-time 
algorithm to determine \emph{Tomaszewski's constant}
to within an additive error of $\pm \eta$; this is the
minimum (over all origin-centered hyperplanes $H$) fraction of points
in $\{-1,1\}^n$ that lie within Euclidean distance $1$ of $H$.
Tomaszewski's constant is conjectured to be $1/2$; lower bounds on it
have been given by Holzman and Kleitman \cite{HK92} and
 independently by Ben-Tal, Nemirovski and Roos 
\ifnum\short=0
\cite{BNR02}.
\fi
\ifnum\short=1
\cite{BNR02short}.
\fi
Our algorithms combine tools from anti-concentration
of sums of independent random variables, Fourier analysis, and Hermite
analysis of linear threshold functions.

\end{enumerate}
\end{abstract}

\newpage

\pagestyle{plain}

\section{Introduction}

This paper is inspired by a belief that simple mathematical objects should
be well understood.  We study two closely related kinds of simple
objects:  $n$-dimensional linear threshold functions
$f(x) = \sign(w\cdot x - \theta)$, and $n$-dimensional origin-centered
hyperplanes $H = \{x \in \R^n : w \cdot x = 0\}.$
Benjamini, Kalai and Schramm \cite{BKS:99} and Tomaszewski \cite{G86}
have posed the question of determining
two universal constants related to halfspaces and origin-centered hyperplanes
respectively; we refer to these quantities as ``the BKS constant'' and
``Tomaszewski's constant.''  While these constants
arise in various contexts including uniform-distribution learning
and optimization theory, little progress has been made on
determining their actual values over the past twenty years.
In both cases there is an easy upper bound which is conjectured to
be the correct value; Gotsman and Linial \cite{GL:94} gave the
best previously known lower bound on the BKS constant in 1994, and
Holzmann and Kleitman \cite{HK92}
gave the best known lower bound on Tomaszewski's constant in 1992.

We give two main results.  The first of these is an improved
lower bound on the BKS constant; a key ingredient
in the proof is a ``robust'' version of the well-known
Khintchine inequality, which we believe may be of
independent interest.  Our second main result is a pair of
algorithms for computing the BKS constant and Tomaszewski's constant
up to any prescribed accuracy.  The first algorithm,
given any $\eta > 0$, runs in time $2^{\poly(1/\eta)}$ and computes
the BKS constant up to an additive $\eta,$ and the second
algorithm runs in time \newa{$2^{{\poly(1/\eta)}}$} and has the same
\ifnum\short=1
guarantee
\fi
\ifnum\short=0
performance guarantee
\fi
for Tomaszewski's constant.

\subsection{Background and problem statements}~

\medskip

\noindent {\bf First problem:  low-degree Fourier weight of
linear threshold functions.}
A \emph{linear threshold function}, henceforth denoted simply LTF, is
a function $f: \{-1,1\}^n \to \{-1,1\}$ of the form
$f(x) = \sign(w \cdot x - \theta)$ where $w \in \R^n$ and $\theta \in \R$
(the univariate function $\sgn: \R \to \R$ is $\sign(z)=1$ for $z \geq 0$
and $\sign(z)=-1$ for $z<0$).
The values $w_1,\dots,w_n$ are the \emph{weights} and $\theta$ is
the \emph{threshold.}
Linear threshold functions play a central role in many areas of computer
science such as concrete complexity theory and machine learning, see
e.g. \cite{DGJ+:10} and the references therein.

It is well known \cite{BKS:99,Peres:04}
that LTFs are highly {noise-stable}, and hence
they must have a large amount of Fourier weight at low degrees.
For $f: \bn \to \R$ and $k \in [0, n]$ let us
define $\w^k[f] = \littlesum_{S \subseteq [n], |S| = k} \wh{f}^2(S)$ and
$\w^{\leq k}[f] = \littlesum_{j=0}^k\w^j[f]$;
we will be particularly interested in the Fourier weight of LTFs
at levels 0 and 1.
More precisely, for $n \in \N$ let $\ltf_n$ denote the set of all
$n$-dimensional LTFs, and let $\ltf = \cup_{n=1}^{\infty} \ltf_n$.
We define the following universal constant:

\begin{definition} \label{def:w1ltf}
\ifnum\short=0
$\w^{\leq 1}[\ltf] \eqdef  \inf_{h \in \ltf} \w^{\leq 1}(h)
= \inf_{n \in \N} \w^{\leq 1}[\ltf_n],$ where
$\w^{\leq 1}[\ltf_n] \eqdef  \inf_{h \in \ltf_n} \w^{\leq 1}(h).$
\fi
\ifnum\short=1
Let
$\w^{\leq 1}[\ltf] \eqdef  \inf_{h \in \ltf} \w^{\leq 1}(h) = \inf_{n \in \N}
\w^{\leq 1}[\ltf_n],$ where $\w^{\leq 1}[\ltf_n] \eqdef  \inf_{h \in \ltf_n}
\w^{\leq 1}(h).$
\fi

\end{definition}

Benjamini, Kalai and Schramm
(see \cite{BKS:99}, Remark 3.7)
and subsequently O'Donnell
(see the Conjecture following Theorem~2 of Section~5.1
of \cite{ODonnellbook}) have conjectured that
$\w^{\leq 1}[\ltf] = 2/\pi$, and hence we will sometimes
refer to
$\w^{\leq 1}[\ltf]$ as ``the BKS constant.''
As $n \to \infty$, a standard analysis of
the $n$-variable Majority function shows that
$\w^{\leq 1}[\ltf] \leq 2/\pi$.  Gotsman and Linial \cite{GL:94} observed
that
$\w^{\leq 1}[\ltf] \geq 1/2$ but until now no better lower bound
was known.
We note that since the universal constant
$\w^{\leq 1}[\ltf]$ is obtained by taking the infimum over an
infinite set, it is not \emph{a priori} clear whether the
computational problem of computing or even approximating
$\w^{\leq 1}[\ltf]$ is decidable.

Jackson
\ifnum\short=0
\cite{Jackson:06}
\fi
\ifnum\short=1
\cite{Jackson:06short}
\fi
has shown that
improved lower bounds on
$\w^{\leq 1}[\ltf]$ translate directly into improved
noise-tolerance bounds for agnostic weak learning
of LTFs in the ``Restricted Focus of Attention'' model of
Ben-David and Dichterman \cite{BenDavidDichterman:98}.
Further motivation for studying $\w^{\le 1} [f]$ comes from
the  fact that $\w^{1} [f]$ is closely related to the
noise stability of $f$ (see \cite{ODonnellbook}).
In particular, if $\mathbf{NS}_{\rho}[f]$ represents the noise
stability of $f$  when the noise rate is $(1-\rho)/2$, then
it is known that
\ifnum\short=0
$$
\left. \frac{d \mathbf{NS}_{\rho}[f]}{d\rho}\right|_{\rho=0}  =\w^{1} [f].
$$
\fi
\ifnum\short=1
$
\left. \frac{d \mathbf{NS}_{\rho}[f]}{d\rho}\right|_{\rho=0}  =\w^{1} [f].
$
\fi
This means that for a function $f$ with $\mathbf{E} [f]=0$, we have
$\mathbf{NS}_{\rho}[f] \rightarrow \rho   \cdot \w^{\le 1}[f] $ as
$\rho \rightarrow 0$. Thus, at very large noise rates,
$\w^{1} [f]$ quantifies the size of the ``noisy boundary" of the mean-zero
function $f$.

\medskip

\noindent
{\bf Second problem:  how many hypercube points have distance
at most 1 from an origin-centered hyperplane?}
For $n \in \mathbb{N}$ and $n >1$, let $\mathbb{S}^{n-1}$ denote the
$n$-dimensional sphere $\mathbb{S}^{n-1} =
\{w  \in \R^n : \| w \|_2=1\}$,
and let $\mathbb{S} = \cup_{n>1} \mathbb{S}^{n-1}$.
Each unit vector $w \in \mathbb{S}^{n-1}$
defines an origin-centered hyperplane
$H_w = \{x \in \R^n: w \cdot x = 0\}.$
Given a unit vector $w \in \mathbb{S}^{n-1}$, we define $\T(w) \in [0,1]$ to be
$\T(w) = \Pr_{x \in \bn} [|w \cdot x| \le 1]$,
the fraction of hypercube points in $\bn$ that lie within
Euclidean distance 1 of the hyperplane $H_w.$
We define the following universal constant, which we call
``Tomaszewski's constant:''

\begin{definition} \label{def:tom}
\ifnum\short=0
$
\T(\mathbb{S}) \eqdef
\inf_{w \in \mathbb{S}} \T(w) =
\inf_{n  \in \mathbb{N} } \T(\mathbb{S}^{n-1}),$
where
$
\T(\mathbb{S}^{n-1}) \eqdef \inf_{w \in \mathbb{S}^{n-1}}
\T(w).
$
\fi
\ifnum\short=1
Define
$
\T(\mathbb{S}) \eqdef
\inf_{w \in \mathbb{S}} \T(w) =
\inf_{n  \in \mathbb{N} } \T(\mathbb{S}^{n-1}),$
where
$\T(\mathbb{S}^{n-1})$  $\eqdef$  $\inf_{w \in \mathbb{S}^{n-1}} \T(w).$
\fi
\end{definition}

Tomaszewski \cite{G86} has conjectured that
$\T(\mathbb{S}) = 1/2$.  The main result of Holzman and Kleitman \cite{HK92}
is a proof that $3/8 \le \T(\mathbb{S})$; the upper bound
$\T(\mathbb{S}) \le 1/2$ is witnessed by
the vector $w=(1/\sqrt{2},1/\sqrt{2}).$
As noted in \cite{HK92}, \newa{the quantity $\T(\mathbb{S})$} has a number of appealing geometric and
probabilistic reformulations.
Similar to the BKS constant, since
$\T(\mathbb{S})$ is obtained by taking the infimum over an infinite set,
it is not immediately evident that any algorithm can compute or approximate
$\T(\mathbb{S})$.
\footnote{Whenever we speak of ``an algorithm to compute or approximate''
one of these constants, of course what we really mean is an
algorithm that outputs the desired value \emph{together with a proof
of correctness of its output value}.}

An interesting quantity in its own right, Tomaszewski's
constant also arises in a range of contexts in optimization theory, see e.g.
\ifnum\short=0
\cite{So09,BNR02}.
\fi
\ifnum\short=1
\cite{So09,BNR02short}.
\fi
In fact, the latter paper proves a lower bound of $1/3$
on the value of Tomaszewski's constant independently of
\cite{HK92}, and independently conjectures
that the optimal lower bound is  $1/2$.

\ignore{
As noted by Holzman and Kleitman, this problem has many appealing
reformulations:  one of these is to determine the minimum
fraction of hypercube $\{-1,1\}^n$ points that lie
within Euclidean distance 1 of any origin-centered hyperplane.
}

\subsection{Our results}~

\smallskip

\noindent {\bf A better lower bound for the BKS constant
$\w^{\leq 1}[\ltf].$}  Our first main result is
the following theorem:

\begin{theorem}[Lower Bound for the BKS constant] \label{thm:w1}
There exists a universal constant $c'>0$  such that
$\w^{\leq 1}[\ltf] \geq {\frac 1 2} + c'.$
\end{theorem}

This is the first improvement on the \cite{GL:94} lower bound of $1/2$
since 1994.  We actually give two quite different proofs of this
theorem, which are sketched in the ``Techniques'' subsection below.

\medskip

\noindent {\bf An algorithm for approximating
the BKS constant $\w^{\leq 1}[\ltf].$}
Our next main result shows that in fact there \emph{is} a finite-time
algorithm that approximates the BKS constant up to any
desired accuracy:

\begin{theorem}[Approximating the BKS constant] \label{thm:w1-algo}
There is an algorithm that, on input an accuracy parameter $\eps>0$,
runs in time $2^{\poly(1/\eps)}$ and outputs a value $\Gamma_\eps$ such that
\begin{equation} \label{eqn:approx-w1}
\w^{\leq 1}[\ltf]   \leq \Gamma_\eps \leq  \w^{\leq 1}[\ltf]+\eps.
\end{equation}
\end{theorem}

\noindent {\bf An algorithm for approximating
Tomaszewski's constant $\T(\mathbb{S}).$}
Our final main result is
\ifnum\short=0
a similar-in-spirit
\fi
\ifnum\short=1
an
\fi
algorithm that approximates $\T(\mathbb{S})$ up to any desired
accuracy:

\begin{theorem}[Approximating Tomaszewski's constant] \label{thm:algo-T}
There is an algorithm that, on input
\ifnum\short=0
an accuracy parameter
\fi
$\eps>0$,
runs in time \newa{$2^{{\poly(1/\eps)}}$} and outputs a value $\Gamma_\eps$ such that
\begin{equation} \label{eqn:approx-Tconst}
{\T(\mathbb{S})}   \leq \Gamma_\eps \leq  {\T(\mathbb{S})}+\eps.
\end{equation}
\end{theorem}

\ifnum\short=0
\subsection{Our techniques for Theorem~\ref{thm:w1}:  lower-bounding
the BKS constant $\w^{\leq 1}[\ltf]$}
\fi
\ifnum\short=1
\subsection{Our techniques for Theorem~\ref{thm:w1}:  lower-bounding
the BKS constant}
\fi

It is easy to show that
it suffices to consider the level-1 Fourier weight
$\w^{1}$ of LTFs that have threshold $\theta=0$
and have $w \cdot x \neq 0$ for all $x \in \{-1,1\}^n$,
so we confine our
discussion to such zero-threshold LTFs (\newa{see Fact~\ref{fac:balanced-w1} for a proof}).
To explain our approaches to lower bounding $\w^{\leq 1}[\ltf]$,
we recall the essentials of
\ifnum\short=0
Gotsman and Linial's simple argument
\fi
\ifnum\short=1
the simple argument of \cite{GL:94}
\fi
that gives a lower bound of $1/2.$  The key ingredient of their argument
is the well-known Khintchine inequality from
functional analysis:

\begin{definition} \label{def:kk}
For a unit vector $w \in \mathbb{S}^{n-1}$
we define
$$\kk(w) \eqdef \E_{x \in \bn} \left[  |w \cdot x| \right]$$
to be the ``Khintchine constant for $w$.''
\end{definition}

\noindent The following is a classical theorem in functional analysis
(we write $e_i$ to denote the unit vector in $\R^n$ with a 1 in
coordinate $i$):

\begin{theorem}[Khintchine inequality, \cite{Szarek:76}] \label{thm:kk}
For $w \in \mathbb{S}^{n-1}$ any unit vector, we have
$ \kk(w) \geq 1/\sqrt{2}$, with equality holding if and only if
$w = {\frac 1 {\sqrt{2}}} \left(\pm e_i \pm e_j\right)$ for some $i \neq j \in [n].$
\end{theorem}
Szarek \cite{Szarek:76} was the first to obtain the optimal
constant $1/\sqrt{2}$, and subsequently several simplifications
of his proof were given \cite{Haagerup82,Tomaszewski-Khintchine,LO94};
we shall give a simple self-contained proof
in Section~\ref{ssec:fourier-k} below. \newa{This proof has previously appeared in
\cite{Garling, Filmus-khintchine} and is essentially a
translation of the \cite{LO94} proof into ``Fourier language.''}
With Theorem~\ref{thm:kk} in hand, the Gotsman-Linial lower bound
is almost immediate:

\begin{proposition}[\cite{GL:94}] \label{prop:gl}
Let $f: \bn \to \bits$ be a zero-threshold LTF
$f(x) = \sgn(w \cdot x)$ where $w \in \R^n$ has $\| w \|_2=1$.
Then $\w^{1}[f] \geq \left( \kk(w) \right)^2.$
\end{proposition}

\begin{proof} We have that
\ifnum\short=0
$$
\kk(w) = \E_x[f(x) (w \cdot x)] = \littlesum_{i=1}^n \wh{f}(i) w_i \leq \sqrt{\littlesum_{i=1}^n \wh{f}^2(i)} \cdot \sqrt{\littlesum_{i=1}^n w_i^2}  = \sqrt{\w^1[f]}
$$
\fi
\ifnum\short=1
$
\kk(w) = \E_x[f(x) (w \cdot x)] =$ $\littlesum_{i=1}^n \wh{f}(i) w_i \leq$
$\sqrt{\littlesum_{i=1}^n \wh{f}^2(i)} \cdot \sqrt{\littlesum_{i=1}^n w_i^2} =$
$\sqrt{\w^1[f]}$
\fi
where the first equality uses the definition of $f$, the second is
Plancherel's identity, the inequality is Cauchy-Schwarz,
and the last equality uses the assumption that $w$ is a unit vector.
\ifnum\short=1
\qed
\fi
\end{proof}

\noindent {\bf First proof of Theorem~\ref{thm:w1}:
A ``robust'' Khintchine inequality.}
Given the strict condition required for equality in the Khintchine
inequality, it is natural to expect that if a unit vector $w \in \R^n$
is ``far'' from ${\frac 1 {\sqrt{2}}}\left(\pm e_i \pm e_j\right)$,
then $\kk(w)$ should be significantly larger than $1/\sqrt{2}$.
We prove a robust version of the Khintchine inequality which
makes this intuition precise.  Given a unit vector $w \in \mathbb{S}^{n-1}$,
define $d(w)$ to be $d(w) = \min \|w - w^*\|_2,$ where $w^*$ ranges
over all $4 {n \choose 2}$ vectors of the form
${\frac 1 {\sqrt{2}}} (\pm e_i \pm e_j).$  Our ``robust Khintchine''
inequality is the following:

\begin{theorem}[Robust Khintchine inequality] \label{thm:robust-kk}
There exists a universal constant $c>0$ such that for any
$w \in \mathbb{S}^{n-1}$, we have
\ifnum\short=0
$$ \kk(w) \geq {\frac 1 {\sqrt{2}}} + c \cdot d(w).$$
\fi
\ifnum\short=1
$ \kk(w) \geq {\frac 1 {\sqrt{2}}} + c \cdot d(w).$
\fi
\end{theorem}

Armed with our robust Khintchine inequality, the simple proof
of Proposition~\ref{prop:gl} suggests a natural approach to lower-bounding
$\w^{\leq 1}[\ltf].$  If $w$ is such that $d(w)$ is ``large'' (at least
some absolute constant), then the statement of Proposition~\ref{prop:gl}
immediately gives a lower bound better than $1/2.$  So the only
remaining vectors $w$ to handle are highly constrained vectors
which are almost exactly of the form ${\frac 1 {\sqrt{2}}} (\pm e_i \pm e_j)$.
A natural hope is that the Cauchy-Schwarz inequality
in the proof of Proposition~\ref{prop:gl} is not tight for such
highly constrained vectors, and indeed this is essentially how we proceed \newa{(modulo some
simple cases in which it is easy to bound $\w^{\leq 1}$ above $1/2$ directly).}
\ignore{we are able to show that this
is the case (though some work is required
for this, see Lemma~\ref{lem:w1-close}).
}

\medskip

\noindent {\bf Second proof of Theorem~\ref{thm:w1}:
anticoncentration, Fourier analysis of LTFs, and LTF approximation.}
Our second proof of Theorem~\ref{thm:w1} employs several
sophisticated ingredients from recent work on structural
properties of LTFs \cite{OS11:chow,MORS:10}.  The first of these
ingredients is a result (Theorem~6.1 of~\cite{OS11:chow}) which
essentially says that any LTF $f(x)=\sign(w \cdot x)$
can be perturbed very slightly to another LTF $f'(x) = \sign(w' \cdot x)$
(where both $w$ and $w'$ are unit vectors).
The key properties of this perturbation are
that (i) $f$ and $f'$ are extremely close, differing only
on a tiny fraction of inputs in $\{-1,1\}^n$; but (ii) the linear
form $w' \cdot x$ has some nontrivial ``anti-concentration''
when $x$ is distributed uniformly over $\bn$, meaning that very
few inputs have $w' \cdot x$ very close to 0.

Why is this useful?  It turns out that the anti-concentration of
$w' \cdot x$, together with results on the degree-1 Fourier spectrum
of ``regular'' halfspaces from \cite{MORS:10}, lets us
establish a lower bound on $\w^{\leq 1}[f']$ that is strictly
greater than $1/2$.  Then the fact that  $f$ and $f'$ agree on
almost every input in $\bn$ lets us argue that the original LTF $f$
must similarly have $\w^{\leq 1}[f]$ strictly greater than $1/2.$
Interestingly, the lower bound on
$\w^{\leq 1}[f']$ is proved using the Gotsman-Linial inequality
$\w^{\leq 1}[f'] \geq (\kk(w\new{'}))^2$; in fact, the anti-concentration
of $w' \cdot x$ is combined with ingredients in the simple Fourier proof
of the (original, non-robust) Khintchine inequality (specifically,
an upper bound on the total influence of the function $\ell(x) =
|w' \cdot x|$) to obtain the result.
\ifnum\short=1
\newa{Because of space constraints we give this second proof in
the full version of the paper, following the references.}
\fi

\ifnum\short=0
\subsection{Our techniques for Theorem~\ref{thm:w1-algo}:  approximating
the BKS constant $\w^{\leq 1}[\ltf]$}\label{ssec:w1algo-techniques}
\fi
\ifnum\short=1
\subsection{Our techniques for Theorem~\ref{thm:w1-algo}:  approximating
the BKS constant}\label{ssec:w1algo-techniques}
\fi

As in the previous subsection, it suffices to consider only zero-threshold
LTFs $\sign(w \cdot x)$.
Our algorithm turns out to be very simple (though its analysis is not):

\medskip

Let $K = \Theta(\eps^{-24}).$ Enumerate all $K$-variable zero-threshold
LTFs, and output the value
\ifnum\short=0
$$\Gamma_\eps \eqdef \min\{\w^1[f] : f
\text{~is a zero-threshold $K$-variable LTF.} \}.$$
\fi
\ifnum\short=1
$\Gamma_\eps \eqdef \min\{\w^1[f] : f
\text{~is a zero-threshold $K$-variable LTF.} \}.$
\fi

\medskip

It is well known (see e.g. \cite{MaassTuran:94})
that there exist $2^{\Theta(K^2)}$ distinct $K$-variable LTFs,
and it is straightforward to confirm that they can be enumerated in
\newa{time $2^{O(K^2 \log K)}$.}
\newa{Since $\w^1[f] $ can be computed in time $2^{O(K)}$ for any given $K$-variable LTF $f$, the above simple algorithm
runs in time $2^{\poly(1/\eps)}$}; the challenge is to show
that the value $\Gamma_\eps$ thus obtained indeed satisfies
Equation~(\ref{eqn:approx-w1}).

A key ingredient in our analysis is the notion of the ``critical index''
of an LTF $f$.  The critical index was implicitly introduced and used in
\cite{Servedio:07cc} and was explicitly used in
\ifnum\short=0
\cite{DiakonikolasServedio:09,DGJ+:10,OS11:chow,DDFS:12stoc}
\fi
\ifnum\short=1
\cite{DiakonikolasServedio:09short,DGJ+:10,OS11:chow,DDFS:12stoc}
\fi
and other works. To define the critical index we need to first define ``regularity'':

\begin{definition}[regularity]
Fix any real value $\tau > 0.$  We say that a vector $w = (w_1, \ldots, w_n) \in \R^n$ is \emph{$\tau$-regular} if $\max_{i \in [n]} |w_i| \leq \tau \|w\| = \tau \sqrt{w_1^2 + \cdots + w_n^2}.$
A linear form $w \cdot x$ is said to be $\tau$-regular if $w$ is $\tau$-regular, and similarly
an LTF is said to be $\tau$-regular
if it is of the form $\sgn(w \cdot x  -\theta)$  where $w$ is $\tau$-regular.
\end{definition}

Regularity is a helpful notion because if $w$ is $\tau$-regular then the Berry-Ess{\'e}en theorem tells us that for uniform $x \in \{-1,1\}^n$, the linear form $w \cdot x$ is ``distributed like a Gaussian up to error $\tau$.''  This can be useful for many reasons (as we will see below).

Intuitively, the critical index of $w$ is the first index $i$ such that
from that point on, the vector $(w_i,w_{i+1},\dots,w_n)$ is regular.  A precise definition follows:

\begin{definition}[critical index]
Given a vector $w \in \R^n$ such that $|w_1| \geq \cdots \geq |w_n| > 0$,
for $k \in [n]$ we denote by $\sigma_k$ the quantity $\sqrt{\littlesum_{i=k}^n w_i^2}$.
We define the \emph{$\tau$-critical index $c(w, \tau)$ of $w$} as the smallest
index $i \in [n]$ for which $|w_i| \leq \tau \cdot \sigma_i$. If
this inequality does not hold for any $i \in [n]$, we define $c(w, \tau) = \infty$.
\end{definition}

Returning to Theorem~\ref{thm:w1-algo}, since our algorithm minimizes over
a proper subset of all LTFs, it suffices to show that
for any zero-threshold LTF $f=\sign(w \cdot x)$, there is a $K$-variable zero-threshold
LTF $g$ such that
\ifnum\short=0
\begin{equation} \label{eqn:w1-goal}
\w^1[g] - \w^1[f]<\eps.
\end{equation}
\fi
\ifnum\short=1
$\w^1[g] - \w^1[f]<\eps.$
\fi
At a high level our proof is a
case analysis based on the size of the $\delta$-critical index $c(w,\delta)$ of
the weight vector $w$, where we choose the parameter $\delta$ to be
$\delta = \poly(\eps)$.
The first case is relatively easy:
if the $\delta$-critical index is large, then it is known that the function
$f$ is very close to some $K$-variable LTF $g$.
Since the two functions agree almost everywhere,
it is easy to show that $|\w^{1}[f] - \w^{1}[g]| \leq \eps$ as desired.

The case that the critical index is small is much more challenging.
In this case it is by no means true that $f$ can be well approximated by an
LTF on few variables -- consider, for example, the majority function.
We deal with this challenge by developing a novel \emph{variable
reduction technique} which lets us construct a $\poly(1/\eps)$-variable
LTF $g$ whose level-1 Fourier weight closely matches that of $f$.

How is this done?
The answer again comes from
the critical index.  Since the critical
index $c(w,\delta)$ is small, we know that except for the
``head'' portion $\sum_{i=1}^{c(w,\delta)-1} w_i x_i$ of the
linear form, the ``tail'' portion
$\sum_{i=c(w,\delta)}^{n} w_i x_i$ of the
linear form ``behaves like a Gaussian.''
Guided by this intuition, our variable reduction technique proceeds
in three steps.
In the first step, we replace the tail coordinates $x_T=
(x_{c(w,\delta)},\dots,x_n)$ by independent
Gaussian random variables and show that the degree-$1$ Fourier
weight of the corresponding ``mixed'' function
(which has some $\pm 1$-valued inputs and some Gaussian inputs)
is approximately equal to $\w^1[f]$.
In the second step, we replace the tail random variable $w_T \cdot G_T$,
where $G_T$ is the vector of Gaussians from the first step,
by a {\em single} Gaussian random variable
$G$, where $G \sim \mathcal{N}(0, \|w_T\|^2).$
We show that this transformation exactly preserves the degree-$1$ weight.
At this point we have reduced the number of variables from $n$ down to $c(w,\delta)$ (which is small in this case!), but the
last variable is Gaussian rather than Boolean.
As suggested by the Central Limit Theorem, though,
one may try to replace this Gaussian random variable by a
normalized sum of independent $\pm 1$ random variables
$\littlesum_{i=1}^M z_i/\sqrt{M}$.
This is exactly the third step of our variable reduction technique.
Via a careful analysis, we show that by taking $M = \poly(1/\eps)$,
this operation preserves the degree-$1$ weight up to an additive $\eps$.
Combining all these steps, we obtain the desired result.

\ifnum\short=0
\subsection{Our techniques for Theorem~\ref{thm:algo-T}:  approximating
Tomaszewski's constant $\T(\mathbb{S})$}
\fi
\ifnum\short=1
\subsection{Our techniques for Theorem~\ref{thm:algo-T}:  approximating
Tomaszewski's constant}
\fi

The first step of our proof of Theorem~\ref{thm:algo-T} is similar in
spirit to
the main structural ingredient of our proof of Theorem~\ref{thm:w1-algo}:
we show
\ifnum\short=0
(Theorem~\ref{thm:T-dog})
\fi
that given any $\eps > 0$, there is a value
$K_\eps=\poly(1/\eps)$ such that
it suffices to consider linear forms $w \cdot x$ over
$K_\eps$-dimensional space, i.e.
for any $n \in \N$ we have
\ifnum\short=0
\[
\T(\mathbb{S}^{n-1}) \leq \T(\mathbb{S}^{K_\eps - 1}) \leq
\T(\mathbb{S}^{n-1}) + \eps.
\]
\fi
\ifnum\short=1
$
\T(\mathbb{S}^{n-1}) \leq \T(\mathbb{S}^{K_\eps - 1}) \leq
\T(\mathbb{S}^{n-1}) + \eps.
$
\fi
Similar to the high-level outline of Theorem~\ref{thm:w1-algo},
our proof again proceeds by fixing any $w \in \mathbb{S}^{n-1}$
and doing a case analysis based on whether the critical index of
$w$ is ``large'' or ``small.'' However, the technical details of each
of these cases is quite different from the earlier proof.
In the ``small critical index'' case
we employ Gaussian anti-concentration (which is inherited by the
``tail'' random variable $w_T x_T$ since the tail vector $w_T$ is regular),
and in the ``large critical index'' case we use an
anti-concentration result from
\cite{OS11:chow}.

Unlike the previous situation for the BKS constant, at this point more
work remains to be done for approximating Tomaszewski's constant.
While there are only $2^{\poly(1/\eps)}$ many halfspaces over
$\poly(1/\eps)$ many variables and hence a brute-force
enumeration could cover all of them in $2^{\poly(1/\eps)}$ time for
the BKS constant, here we must contend with the fact that
$\mathbb{S}^{K_\eps-1}$ is an uncountably infinite set, so we cannot
naively minimize over all its elements.
Instead we take a dual approach
and exploit the fact that while there are uncountably infinitely
many vectors in $\mathbb{S}^{K_\eps-1}$,
there are only $2^{K_\eps}$ many hypercube points in
$\{-1,1\}^{K_\eps}$, and (with some care) the desired infimum
over all unit vectors can be formulated in the language of existential theory of the reals.
We then use an algorithm for deciding existential theory of the reals (see \cite{Ren:88})
to compute the infimum.
\ifnum\short=1
\newa{Because of space constraints we prove Theorem~\ref{thm:algo-T}
in the full version of the paper, following the references.}
\fi

\medskip

\noindent {\bf Discussion.}
It is interesting to note that determining Tomaszewski's constant is
an instance of the well-studied generic problem of understanding tails of
Rademacher sums. For the sake of discussion,  let us
define $\mathbf{T}_{\mathsf{in}}(w,a)
= \Pr_{x \in \{-1,1\}^n} [|w \cdot x| \le a]$
and $\mathbf{T}_{\mathsf{out}}(w,a)
= \Pr_{x \in \{-1,1\}^n} [|w \cdot x| \ge a]$
where $w \in \mathbb{S}^{n-1}$. Further, let $\mathbf{T}_{\mathsf{in}}(a) =
\inf_{w \in \mathbb{S}} \mathbf{T}_{\mathsf{in}} (w,a)$ and $\mathbf{T}_{\mathsf{out}}(a) =
\inf_{w \in \mathbb{S}} \mathbf{T}_{\mathsf{out}} (w,a)$.
Note that Tomaszewski's constant $\mathbf{T}(\mathbb{S})$
is simply $\mathbf{T}_{\mathsf{in}}(1)$.
Much effort has been expended on getting sharp estimates for
$\mathbf{T}_{\mathsf{in}}(a)$ and $\mathbf{T}_{\mathsf{out}}(a)$ for
various values of $a$ (see e.g.~\cite{Pinelis:12,Bentkus:04}).
As a representative example, Bentkus and Dzindzalieta~\cite{BK:12} proved that
\ifnum\short=0
$$
\mathbf{T}_{\mathsf{in}}(a) \ge \frac14 + \frac14 \cdot \sqrt{2 -
\frac{2}{a^2}}$$
\fi
\ifnum\short=1
$\mathbf{T}_{\mathsf{in}}(a) \ge \frac14 + \frac14 \cdot \sqrt{2 -
\frac{2}{a^2}}$
\fi
 for $a \in (1,\sqrt{2}]$.
Similarly, Pinelis~\cite{Pinelis:94} showed that there is an absolute constant $c>0$ such that
$
\mathbf{T}_{\mathsf{out}}(a) \ge 1 - c \cdot \frac{\phi(a)}{a}
$
where $\phi(x)$ is the density function of the standard normal
$\mathcal{N}(0,1)$ (note this beats the standard Hoeffding bound
by a factor of $1/a$).

On the complementary side, Montgomery-Smith~\cite{Montgomery-Smith:90}
proved that there is an absolute constant $c'>0$ such that
$
\mathbf{T}_{\mathsf{out}}(a) \ge e^{-c' \cdot a^2}
$
\newb{for all $a \leq 1.$}
Similarly, Oleszkiewicz \cite{Olek:96} proved that $\mathbf{T}_{\mathsf{out}}(
1) \ge 1/10$. The conjectured lower bound on $\mathbf{T}_{\mathsf{out}}(1)$ is
$7/32$ (see \cite{HK:93}).  While we have not investigated
this in detail, we suspect that our techniques may
be applicable to some of the above problems.  Finally, we note that
apart from being of intrinsic interest to functional analysts and probability
theorists, the above quantities arise frequently in the
optimization literature
\ifnum\short=0
(see \cite{HLNZ:07,BNR02}).
\fi
\ifnum\short=1
(see \cite{HLNZ:07short,BNR02short}).
\fi
Related tail bounds have
also found applications in extremal combinatorics
\ifnum\short=0
(see \cite{AHS:12}).
\fi
\ifnum\short=1
(see \cite{AHS:12short}).
\fi

\ifnum\short=1

\medskip

\noindent {\bf Mathematical Preliminaries.}
These are given in Section~2 of the full version.

\fi


\section{Mathematical Preliminaries} \label{sec:prelims}

\subsection{Fourier analysis over $\{-1,1\}^n$ and influences} \label{ssec:fourier}

We consider functions $f : \bn \to \R$ (though we often focus on
Boolean-valued functions which map to $\{-1,1\}$), and we think of
the inputs $x$ to $f$ as being distributed according to the uniform
probability distribution. The set of such functions forms a
$2^n$-dimensional inner product space with inner product given by
$\la f, g \ra = \E_x[f(x)g(x)]$. The set of functions $(\chi_S)_{S
\subseteq [n]}$ defined by $\chi_S(x) = \prod_{i \in S} x_i$ forms a
complete orthonormal basis for this space.  We will also often write
simply $x_S$ for $\prod_{i \in S} x_i$.  Given a function $f : \bn
\to \R$ we define its \emph{Fourier coefficients} by $\wh{f}(S) =
\E_x[f(x) x_S]$, and we have that $f(x) = \sum_S \wh{f}(S) x_S$.

As an easy consequence of orthonormality we have \emph{Plancherel's
identity} $\la f, g \ra = \sum_S \wh{f}(S) \wh{g}(S)$, which has as
a special case \emph{Parseval's identity}, $\E_x[f(x)^2] = \sum_S
\wh{f}(S)^2$. From this it follows that for every $f : \bn \to
\bits$ we have $\sum_S \wh{f}(S)^2 = 1$.
Note that for $f:\bn \to \R$ we have that $\Var[f]  = \E_x[f^2(x)] - \left(\E_x[f]\right)^2 = \littlesum_{S \neq \emptyset} \wh{f}^2(S).$


\begin{definition}  \label{def:influences}
Given $f : \bn \to \R$ and $i \in [n]$, the \emph{influence} of variable $i$
is defined as $\Inf_i(f) = \E_x \left[ \Var_{x_i} [f(x)] \right]$.
The {\em total influence} of $f$ is defined as $\Inf(f) = \littlesum_{i=1}^n \Inf_i(f).$
\end{definition}

\begin{fact} \label{fact:influences}
We have the identity $\Inf_i(f) = \littlesum_{S \ni i} \wh{f}^2(S)$; moreover, for $f : \bn \to \bits$ (i.e. Boolean-valued),
it holds $\Inf_i(f) = \Pr_x[f(x\supim) \neq f(x\supip)]$, where $x\supim$ and $x\supip$
denote  $x$ with the $i$'th bit set to $-1$ or $1$ respectively. If $f:\bn \to \bits$ is unate, then $\Inf_i(f) = |\hat{f}(i)|$.
\end{fact}

\begin{fact} \label{fact:ordering}
Let $f  = \sgn(\littlesum_{i=1}^n w_ix_i  - w_0)$ be an LTF such that $|w_1| \geq |w_i|$ for all
$i \in [n]$.  Then $|\Inf_1(f)| \geq |\Inf_i(f)|$ for all $i \in [n]$. Moreover, for all $i \in [n]$ it holds $w_i \cdot \wh{f}(i) \geq 0.$
\end{fact}

\ignore{

\inote{Do we actually need the following two definitions?}

\begin{definition}
For $a,b \in \R$ we write $a \apx{\eta} b$ to indicate that $|a - b|
\leq O(\eta)$.
\end{definition}

\begin{definition} A function $f : \bits^n \to \R$ is said to be
a ``junta on $J \subset [n]$'' or a ``$J$-junta'' if $f$ only depends on the
coordinates in $J$.
\end{definition}

}

\subsection{Probabilistic Facts} \label{ssec:prob-basics}

We require some basic probability results including the standard additive Hoeffding bound:
\begin{theorem} \label{thm:chb}
Let $X_1, \ldots, X_n$ be independent random variables such that for
each $j \in [n]$, $X_j$ is supported on $[a_j, b_j]$ for some $a_j,
b_j \in \R$, $a_j \le b_j$. Let $X \ = \littlesum_{j=1}^{n} X_j$.
Then, for any $t>0$,
$\Pr \big[ |X - \E[X]| \ge t \big] \le 2 \exp \left(
-2t^2/\littlesum_{j=1}^{n} (b_j-a_j)^2   \right).$
\end{theorem}

\medskip

\noindent The Berry-Ess{\'e}en theorem (see e.g. \cite{Feller})
gives explicit error bounds for the Central
Limit Theorem:

\begin{theorem} \label{thm:be} (Berry-Ess{\'e}en)
Let $X_1, \dots, X_n$ be independent random variables satisfying
$\E[X_i] = 0$ for all $i \in [n]$, $\sqrt{\littlesum_i \E[X_i^2]} =
\sigma$, and $\littlesum_i \E[|X_i|^3] = \rho_3$.  Let $S = (X_1 +
\cdots + X_n)/\sigma$ and let $F$ denote the cumulative distribution
function (cdf) of $S$. Then
$\sup_x |F(x) - \Phi(x)| \leq \rho_3/\sigma^3$
where $\Phi$ denotes the cdf of the standard gaussian random variable.
\end{theorem}

An easy consequence of the Berry-Ess{\'e}en theorem is the following fact, which says that
a regular linear form has good anti-concentration (i.e. it assigns small probability
mass to any small interval):

\begin{fact} \label{fact:be}
Let $w=(w_1,\dots,w_n)$ be a $\tau$-regular vector in $\R^n$ and write $\sigma$ to denote $\|w\|_2$.
Then for any interval $[a,b] \subseteq \R$, we have
$\big|\Pr[\littlesum_{i=1}^n w_i x_i \in (a,b]] - \Phi([a/\sigma, b/\sigma])\big| \leq 2\tau$,
where $\Phi([c,d]) \eqdef \Phi(d) - \Phi(c)$. In particular, it follows that
$\Pr \big[\littlesum_{i=1}^n w_i x_i \in (a,b]\big] \leq  |b -a| / \sigma + 2\tau.$
\end{fact}

%
%
%
%

\subsection{Technical Tools about Regularity and the Critical Index} \label{ssec:tools}

The following simple fact states that the ``tail weight'' of the vector $w$  decreases exponentially prior to the critical index:
\begin{fact}\label{fact:small-tail}
For any vector $w = (w_1, \ldots, w_n)$ such that $|w_1| \geq \cdots \geq |w_n| > 0$ and $1 \le a \le c(w,\tau)$, we have $\sigma_a < (1-\tau^2)^{(a-1)/2} \cdot \sigma_1$.
\end{fact}
\begin{proof}
If $a <c(w,\tau)$, then by definition $|w_a| > \tau \cdot \sigma_a$. This implies that $\sigma_{a+1} < \sqrt{1-\tau^2} \cdot \sigma_a$. Applying this inequality repeatedly, we get  that $\sigma_a < (1-\tau^2)^{(a-1)/2} \cdot \sigma_1$ for any  $1 \le a \le c(w,\tau)$.
\end{proof}

\noindent We will also need the following corollary
(that appears e.g. as Propositions 31 and 32 in~\cite{MORS:10}).

\begin{fact}\label{fact:gaussian-vs-reg}
Let $\ell(x) = w \cdot x - w_0$ with $\|w\|_2 = 1$ and $w_0 \in \R$ and $f(x) = \sgn(\ell(x))$.
If $w$ is $\tau$-regular, then we have:
\begin{itemize}
\item {(i)} $\E_{x \sim \mathcal{U}_n} [ f(x) ]  \approx^{\tau}  \E_{x \sim \mathcal{N}^n } [f(x)]$ and
\item {(ii)} $\E_{x \sim \mathcal{U}_n} [ |\ell(x)| ]  \approx^{\tau}  \E_{x \sim \mathcal{N}^n} [ |\ell(x)| ]$,
\end{itemize}
\newa{where $\mathcal{N}$ denotes the standard Gaussian distribution $N(0,1)$.}
\end{fact}

\subsection{Miscellaneous}

For $a,b \in \R$ we write $a \apx{\eta} b$ to indicate that $|a - b|
\leq O(\eta)$.

For a vector $w \in \R^n$, we write $w^{(k)}$ to denote the $(n-k)$-dimensional vector
obtained by taking the last $n-k$ coordinates of $w$, i.e.
$w^{(k)} = (w_{k+1},\dots,w_n)$.

We will use the following elementary fact, which is a direct consequence of Cauchy-Schwarz.

\begin{fact} \label{fact:close}
Let $a, b \in \R^m$ with $\|a\|_2 \leq 1$, $\|b\|_2 \leq 1$ such that $\|a-b\|_2^2 \leq \eta$. Then
$$ \left| \|a\|^2_2 -  \|b\|^2_2   \right| \leq 2\sqrt{\eta}.$$
\end{fact}

\begin{proof}
We have that
\begin{eqnarray*}
\left| \littlesum_{i=1}^m {(a_i^2  - b_i^2)} \right|
= \left| \littlesum_{i=1}^m \big( a_i-b_i \big) \big( a_i+b_i \big)  \right|
&\leq& \sqrt{\littlesum_{i=1}^m \big( a_i+b_i \big)^2} \cdot \sqrt{\littlesum_{i=1}^m \big( a_i-b_i \big)^2} \\
&\leq& \sqrt{ 2 \cdot \littlesum_{i=1}^m \big( a_i^2+b_i^2  \big)} \cdot \| a-b \|_2 \leq 2 \sqrt{\eta}
\end{eqnarray*}
where the first inequality is Cauchy-Schwarz, the second uses
the elementary fact $(a+b)^2 \leq 2(a^2+b^2)$, for all $a, b \in \R$, while the third uses
our assumption that $\|a\|_2, \|b\|_2 \leq 1.$
\end{proof}


\section{Proof of Theorem~\ref{thm:robust-kk}:  A ``robust'' Khintchine inequality} \label{sec:robust-proof}

It will be convenient for us to reformulate Theorems~\ref{thm:kk} and~\ref{thm:robust-kk} as follows: Let us say that a unit vector $w = (w_1, \ldots, w_n) \in \mathbb{S}^{n-1}$
is \emph{proper} if $w_i \geq w_{i+1} \geq 0$ for all $i \in [n-1]$.  Then we may state
the ``basic'' Khintchine inequality with optimal constant,
Theorem~\ref{thm:kk}, in the following equivalent way:

\begin{theorem}[Khintchine inequality, \cite{Szarek:76}] \label{thm:kk2}
Let $w \in \R^n$ be a proper unit vector,
so $w_1 \geq \cdots \geq w_n \geq 0.$ 
Then $ \kk(w) \geq 1/\sqrt{2},$ with equality holding if and only if
$w = w^{\ast} \eqdef (1/\sqrt{2}, 1/\sqrt{2}, 0, \ldots, 0).$
\end{theorem}

And we may restate our ``robust'' Khintchine inequality, Theorem~\ref{thm:robust-kk},
as follows:

\begin{theorem}[Robust Khintchine inequality] \label{thm:robust-kk2}
There exists a universal constant $c>0$ such that the following holds:
Let $w \in \R^n$ be a proper unit vector. Then $ \kk(w) \geq 1/\sqrt{2} + c \cdot \|w - w^{\ast}\|_2,$
 where
$w^{\ast} \eqdef (1/\sqrt{2}, 1/\sqrt{2}, 0, \ldots, 0).$
\end{theorem}

Before we proceed with the proof of Theorem~\ref{thm:robust-kk2}, we give a simple Fourier analytic proof of the ``basic'' Khintchine inequality with optimal constant, $\kk(w)
\geq 1/\sqrt{2}$.  \newb{(We note that this is a well-known argument by now;
it is given in somewhat more general form in \cite{Oleszkiewicz:99}
and in \cite{KLO:96}.)}
We then build on this to prove Theorem~\ref{thm:robust-kk2}.

\subsection{Warm-up: simple proof that $\kk(w) \geq 1/\sqrt{2}$} \label{ssec:fourier-k}

We consider the function $\ell(x) = \left|\littlesum_{i=1}^n w_ix_i \right|$ where $\littlesum_i w_i^2=1$ and will show that $\kk(w) = \E_x [\ell(x)] \geq 1/\sqrt{2}.$
Noting that  $\E_x[(\ell(x))^2] =1$, we have $(\mathbf{E} [\ell(x)])^2 = 1 - \Var [\ell]$, so it suffices to  show that $\Var[\ell] \leq 1/2$. This follows directly
by combining the following claims. The first bound is an improved Poincar{\'e} inequality for even functions:

\begin{fact}\label{fac:even}(Poincar{\'e} inequality)
Let $f : \bn \to \R$ be even. Then $\Var[f] \leq (1/2) \cdot \Inf(f).$
\end{fact}
\begin{proof} Since $f$ is even, we have that $\wh{f}(S) = 0$ for all $S$ with odd $|S|$.
We can thus write
\begin{eqnarray*}
\Inf(f)  =  \littlesum_{S \subseteq [n], |S|\textrm{ even}} |S| \cdot 
\wh{f}^2(S) &\geq& 2 \cdot  \littlesum_{\emptyset \neq S \subseteq [n], |S|\textrm{ even}} \wh{f}^2(S)\\
& = & 2 \cdot  \littlesum_{\emptyset \neq S \subseteq [n]} \wh{f}^2(S) = 
2 \cdot \Var[f].
\end{eqnarray*}
\end{proof}

The second is an  upper bound on the influences in $\ell$ as a function of the weights:

\begin{lemma}\label{lem:inf-ell}
Let $\ell(x) = \left|\littlesum_{i=1}^n w_ix_i \right|$. For any $i \in [n]$, we have $\Inf_i(\ell) \leq w_i^2.$
\end{lemma}

\begin{proof}
Recall that $ \Inf_i(\ell) = \E_{x} \left[  \Var_{x_i} \big[ \ell(x) \big]
\right] = \E_{x}\left[
\E_{x_i}[
\ell^2(x)]  - (\E_{x_i}[\ell(x)])^2
\right]$.
We claim that for any $x \in \bn$, it holds that $\Var_{x_i} [ \ell(x) ] \leq w_i^2$, which yields the lemma.
To show this claim we write $\ell(x) = \left|  w_i x_i + c_i \right|$, where $c_i =  \littlesum_{j \neq i} w_j \cdot x_j$ does not depend on $x_i$.
\ignore{We want to show an upper bound on the
quantity $\Var_{x_i} [\ell(x)] = \E_{x_i}[\ell^2(x)]
- (\E_{x_i}[\ell(x)])^2.$}

Since $\ell^2(x) = c_i^2 + w_i^2 + 2c_i w_i x_i$, it follows that $\E_{x_i}[\ell^2(x)] = c_i^2 + w_i^2$, and clearly
$\E_{x_i}[\ell(x)]  = (1/2)\cdot (|w_i - c_i| + |w_i + c_i|).$
We consider two cases based on the relative magnitudes of $c_i$ and $w_i$.

If $|c_i| \leq |w_i|$, we have $\E_{x_i}[\ell(x)] = (1/2)\cdot \left( \sgn(w_i) (w_i - c_i) + \sgn(w_i) (w_i + c_i)\right) = |w_i|$. Hence, in this case
$\Var_{x_i}[\ell(x)] = c_i^2 \leq w_i^2$.
If on the other hand $|c_i| > |w_i|$, then we have $\E_{x_i}[\ell(x)]
= (1/2)\cdot \left( \sgn(c_i) (c_i-w_i) + \sgn(c_i) (c_i+w_i)\right)
= |c_i|$, so again $\Var_{x_i}[\ell(x)] = w_i^2$ as desired.
\ifnum\short=1
\qed
\fi
\end{proof}

The bound $\kk(w) \geq 1/\sqrt{2}$ follows from the above two claims using the fact that $\ell$ is even and that $\littlesum_i w_i^2=1.$

\subsection{Proof of Theorem~\ref{thm:robust-kk2}} \label{ssec:robust-kk}

Let $w\in \R^n$ be a proper unit vector and denote $\tau = \|w-w^{\ast}\|_2$.
To prove Theorem~\ref{thm:robust-kk2}, one would intuitively want to obtain a robust version of the simple Fourier-analytic proof of Theorem~\ref{thm:kk2}
from the previous subsection. Recall that the latter proof boils down to the following:
$$\Var[\ell] \leq (1/2) \cdot \Inf(\ell) = (1/2)\cdot \littlesum_{i=1}^n \Inf_i(\ell) \leq (1/2) \cdot \littlesum_{i=1}^n w_i^2 = 1/2$$
where the first inequality is Fact~\ref{fac:even} and the second is Lemma~\ref{lem:inf-ell}. While it is clear that both inequalities can be individually tight,
one could hope to show that both inequalities cannot be tight simultaneously. It turns out that this intuition is not quite true, however it holds
if one imposes some additional conditions on the weight vector $w$. The remaining cases for $w$
that do not satisfy these conditions can be handled by elementary arguments.

\ignore{


}


We first note that without loss of generality we may assume that \ $w_1 = \max_i w_i > 0.3$, for
otherwise Theorem~\ref{thm:robust-kk2} follows directly
from the following result of K{\"o}nig {\em et al}:

\begin{theorem} [\cite{KST99}] \label{thm:TJ}
For a proper unit vector $w \in \R^n$, we have $\kk(w) \geq \sqrt{{2/\pi}} - (1 - \sqrt{{2/\pi}}) w_1.$
\end{theorem}

Indeed, if $w_1 \leq 0.3$, the above theorem gives that
\ifnum\short=0
$$\kk(w) \geq 1.3 \sqrt{{ 2 / \pi}} - 0.3 > 0.737 > 1/\sqrt{2} + 3/100 \geq1/\sqrt{2} + (1/50)\tau,$$
\fi
\ifnum\short=1
$\kk(w) \geq 1.3 \sqrt{{ 2 / \pi}} - 0.3 > 0.737 > 1/\sqrt{2} + 3/100 \geq1/\sqrt{2} + (1/50)\tau,$
\fi
where the last inequality follows from the fact that $\tau\leq \sqrt{2}$ (as both $w$ and $w^{\ast}$ are unit vectors).
Hence, we will henceforth assume that $w_1 >0.3.$ \newa{(We note that there is nothing
special about the number $0.3$; by adjusting various constants elsewhere in the argument, our proof can be made to work with $0.3$ replaced by any (smaller) absolute positive constant. As a result, we could have avoided using Theorem~\ref{thm:TJ} and used quantitatively weaker versions of the theorem which can be shown to follow easily from the Berry-Ess{\'e}en theorem. However, for convenience we have used
Theorem~\ref{thm:TJ} and the number $0.3$ in what follows.)}

The preceding discussion leads us to the following definition:

\begin{definition}[canonical vector] \label{def:can}
We say that a proper unit vector $w \in \R^n$ is {\em canonical} if it satisfies the following conditions:
\ifnum\short=0
\begin{enumerate}
\item[(a)] $w_1 \in [0.3, 1/\sqrt{2}+ 1/100]$;

\item[(b)]  $\tau = \|w-w^{\ast}\|_2 \geq \newa{2}/5$;

\end{enumerate}
\fi
\ifnum\short=1
(a) $w_1 \in [0.3, 1/\sqrt{2}+ 1/100]$; and
(b)  $\tau = \|w-w^{\ast}\|_2 \geq 1/5$.
\fi
\end{definition}



The following lemma establishes Theorem~\ref{thm:robust-kk2} for non-canonical vectors:

\begin{lemma} \label{lem:non-can}
Let $w$ be a proper non-canonical vector. Then $\kk(w) \geq 1/\sqrt{2}+ (1/1000)\tau$, where $\tau = \|w-w^{\ast}\|_2.$
\end{lemma}

\ifnum\short=0
The proof of Lemma~\ref{lem:non-can} is elementary, using only basic facts about
symmetric random variables, but sufficiently long that we give it
in Section~\ref{sec:non-can-proof}.  For canonical vectors we show:
\fi

\ifnum\short=1
\newa{
The proof of Lemma~\ref{lem:non-can} is elementary, using only
basic facts about symmetric random variables, but sufficiently
long that we give it in the full version.}  For canonical vectors we show:
\fi

\begin{theorem}\label{thm:main-kk} There exist universal constants $c_1, c_2>0$ such that:
	Let $w \in \mathbb{R}^n$ be canonical.
	Consider the mapping $\ell(x) = |w \cdot x|$. Then at least one of the following statements is true :
\ifnum\short=0
	\begin{enumerate}
		\item[(1)] $\Inf_1(\ell) \le w_1^2 - c_1$;
		\item[(2)] $\w^{> 2}[\ell] \geq c_2.$
	\end{enumerate}
\fi
\ifnum\short=1
(1) $\Inf_1(\ell) \le w_1^2 - c_1$; (2) $\w^{> 2}[\ell] \geq c_2.$
\fi
\end{theorem}

This proof is more involved, using Fourier analysis and critical index
\ifnum\short=0
arguments.  We
defer it to Section~\ref{ssec:main-kk}, and
\fi
\ifnum\short=1
arguments \newa{(see the full version)}.  We
\fi
proceed now to show that
for canonical vectors, Theorem~\ref{thm:robust-kk2} follows from Theorem~\ref{thm:main-kk}. To see this we argue as follows:
Let $w \in \R^n$ be canonical. We will show that there exists a universal constant $c>0$ such that
$\kk(w) \geq 1/\sqrt{2} + c$;
as mentioned above, since $\tau < \sqrt{2}$, this is sufficient for our purposes. Now recall that
\begin{equation} \label{eqn:kk-def}
\kk(w) = \E_x[\ell(x)] =  \wh{\ell}(0) =  \sqrt{1-\Var[\ell]}.
\end{equation}
In both cases, we will show that there exists a constant $c'>0$ such that

\begin{equation} \label{eq:varub}
\Var[\ell] \leq 1/2 - c'.
\end{equation}
From this (\ref{eqn:kk-def}) gives
$\kk(w) \geq \sqrt{1/2+c'} = 1/\sqrt{2}+c''$
where $c'' >0$ is a universal constant, so to establish Theorem~\ref{thm:robust-kk2} it suffices
to establish (\ref{eq:varub}).

Suppose first that statement (1) of Theorem~\ref{thm:main-kk} holds. In this case we exploit the fact that Lemma~\ref{lem:inf-ell} is not tight. We can write
$$\Var[\ell] \leq (1/2)\cdot \Inf(f) \leq  (1/2)\cdot \left( w_1^2- c_1 + \littlesum_{i=2}^n w_i^2 \right) \leq (1/2)-c_1/2,$$
giving (\ref{eq:varub}).
Now suppose that statement (2) of Theorem~\ref{thm:main-kk} holds, i.e. at least a $c_2$ fraction of the total Fourier mass of $\ell$ lies above level $2$.
Since $\ell$ is even, this is equivalent to the statement $\w^{\geq 4}[\ell] \geq c_2.$ In this case, we prove a better upper bound on the variance because
Fact~\ref{fac:even} is not tight. In particular, we have
$$\Inf(\ell) \geq 2 \w^{2}[\ell]+4\w^{\geq 4}[\ell]
= 2\left( \Var[\ell] - \w^{\geq 4}[\ell] \right)+4\w^{\geq 4}[\ell] = 2 \Var[\ell] + 2\w^{\geq 4}[\ell]$$
which yields
$ \Var[\ell] \leq (1/2) \Inf(\ell) - \w^{\geq 4}[\ell] \leq (1/2)- c_2,$
again giving (\ref{eq:varub}) as desired.


\ifnum\short=0

\subsection{Proof of Lemma~\ref{lem:non-can}} \label{sec:non-can-proof}

We will need the following important claim for the proof of Lemma~\ref{lem:non-can}.

\begin{claim} \label{claim:sym}
Let $X$ be a symmetric discrete random variable supported on $\R$, i.e.
$\Pr[X=x]=\Pr[X=-x]$ for all $x \in \R.$
Then for all $c \in \R$ we have
\[
\max\{\E[|X|],|c|\} \leq \E[|X + c|].
\]
\end{claim}
\begin{proof}
\newa{Since $X$ is symmetric, $c+X$ and $c-X$ have the same distribution. As a result, we have
$
\E[|X + c|] = (1/2) \cdot (\E[|c+X|] + \E[|c-X|] ).
$ Further,
\begin{eqnarray*}
(1/2) \cdot (\E[|c+X|] + \E[|c-X|] ) &=& \E\big[(1/2) \cdot \big(|c+X| + 
|c-X| \big) \big]\\
&=& \E[\max \{|X|, |c| \} ] \ge \max \{ \E[|X|], c\}
\end{eqnarray*}
which finishes the proof.
}
\end{proof}

\begin{proof}[{\bf Proof of Lemma~\ref{lem:non-can}}]
If $w$ is a non-canonical vector, then there are exactly two possibilities :
\newline
\newline
\textbf{Case 1:} $w_1 \not  \in [0.3, 1/\sqrt{2}+ 1/100]$. In case $w_1 \le 0.3$, then the calculation following Theorem~\ref{thm:TJ} already gives us that $\kk(w)  \geq1/\sqrt{2} + (1/50)\tau$. The other possibility is that $w_1 \ge 1/\sqrt{2} + 1/100$. In this case,
\begin{eqnarray*}
\kk(w) = \mathbf{E} \left[|\littlesum_{i=1}^n w_i x_i|\right]  &=&  
(1/2) \cdot \mathbf{E} \left[|w_1 + \littlesum_{i=2}^n w_i x_i|\right]  +(1/2)\cdot \mathbf{E} \left[|-w_1 + \littlesum_{i=2}^n w_i x_i|\right]\\
& \ge & \frac{|w_1| + |w_1|}{2}  = |w_1|
\end{eqnarray*}
where the inequality is an application of Claim~\ref{claim:sym}.  As $|w_1| \ge 1/\sqrt{2} + 1/100$, we get  that $$\kk(w) \geq 1/\sqrt{2} + 1/100 \geq 1/\sqrt{2} + 1/200 \cdot \tau$$(using that $\tau \le \sqrt{2}$).
\newline
\newline
\textbf{Case 2:} $\tau \le \newa{2}/5$. Of course, here we can also assume
that $w_1 \in  [0.3, 1/\sqrt{2}+ 1/100]$ (since otherwise, \textbf{Case 1}
proves the claim).  We let $w_1 = 1/\sqrt{2} - a$ and $w_2 = 1/\sqrt{2} -b$
and $\sum_{i>2} w_i^2=c^2$.  By definition, we have that $a \le b$ and
$b \ge 0$.  Also,
\begin{equation}\label{eq:dist}
 \tau^2=\| w - w^{\ast} \|_2^2 = a^2 + b^2 + c^2.
\end{equation}
Moreover, since $w$ is a unit vector, we have that
\begin{equation}\label{eq:norm}
a^2 + b^2 + c^2 = \sqrt{2} (a+b).
\end{equation}
Expanding the expression for $\kk(w)$ on $x_1, x_2$ and recalling that
$x^{(2)}=(x_3,\dots,x_n)$, we get
\begin{eqnarray*}
\kk(w) &=& \frac{1}{2} \cdot  \left( \mathbf{E}_{x^{(2)} \in \{-1,1\}^{n-2}}\left[  \left| \sqrt{2} - (a+b) + w^{(2)} \cdot x^{(2)} \right| \right]\right.\\
&& \left.\text{~~~~~~} +  \mathbf{E}_{x^{(2)} \in \{-1,1\}^{n-2}}\left[  \left|  (a-b) + w^{(2)} \cdot x^{(2)} \right| \right] \right) \\
&\ge& \frac{1}{2} \cdot \left( \max \{| \sqrt{2} - (a+b) |, \mathbf{E} [| w^{(2)} \cdot x^{(2)} |] \} + \max \{| a-b |, \mathbf{E} [| w^{(2)} \cdot x^{(2)} |] \}  \right) \\
&\ge& \frac{1}{2} \cdot \left( \max \left\{| \sqrt{2} - (a+b) |, \frac{c}{\sqrt{2}} \right\} + \max \left\{| a-b |, \frac{c}{\sqrt{2}}\right\}  \right)
\end{eqnarray*}
where the first inequality follows from Claim~\ref{claim:sym} and the
second inequality uses the fact $\mathbf{E}[|w^{(2)} \cdot x^{(2)}|] \ge c/\sqrt{2}$ (as follows from Theorem~\ref{thm:kk}).
We consider two further sub-cases :
\smallskip

\noindent \textbf{Case 2(a)}:  Let $c^2 \ge \tau^2/20$.  Then, we can bound the right hand-side from below as follows:
\begin{eqnarray*}
&& \frac{1}{2} \cdot \left( \max \left\{| \sqrt{2} - (a+b) |,
\frac{c}{\sqrt{2}} \right\} + \max \left\{| a-b |, \frac{c}{\sqrt{2}}\right\}
\right) \ge \frac{1}{2} \cdot \left( | \sqrt{2} - (a+b) | + \frac{c}{\sqrt{2}}
 \right) \\ &\ge& \frac{1}{2} \left( \left|\sqrt{2} -
\frac{\tau^2}{\sqrt{2}}\right| \right) + \frac{\tau}{4 \sqrt{10}}  =
\frac{1}{\sqrt{2}} - \frac{\tau^2}{2\sqrt{2}} +  \frac{\tau}{\newa{\sqrt{40}}}
\end{eqnarray*}
where the second inequality uses (\ref{eq:norm}).  As long as $\tau \le
\newa{2/5}$, it is easy to check that
$$
\frac{\tau}{\newa{\sqrt{40}}}  - \frac{\tau^2}{2\sqrt{2}} \ge \frac{\tau}{1000}
$$
which proves the assertion in this case.

\smallskip

\noindent \textbf{Case 2(b)}: Let $c^2 < \tau^2/20$. In this case, we will prove a lower bound on $|a-b|$.   Using $c^2 <\tau^2/20$ and (\ref{eq:dist}), we have $a^2 + b^2 > (19 \tau^2)/20$.   Also, using (\ref{eq:norm}), we have $a + b = \tau^2/\sqrt{2}$.
We now have
\begin{eqnarray*}
(a-b)^2 = 2(a^2 + b^2) - (a+b)^2  \ge 2 \cdot \frac{19}{20} \cdot \tau^2  - \frac{\tau^4}{2}  \ge \frac{19}{10} \tau^2 - \frac{\tau^4}{2}  \ge \tau^2
\end{eqnarray*}
The last inequality uses $\tau \le \newa{2/5}$. Now, as in \textbf{Case 2(a)},
we have \begin{eqnarray*}
&& \frac{1}{2} \cdot \left( \max \left\{| \sqrt{2} - (a+b) |,
\frac{c}{\sqrt{2}} \right\} + \max \left\{| a-b |, \frac{c}{\sqrt{2}}\right\}
\right) \ge \frac{1}{2} \cdot \left( | \sqrt{2} - (a+b) | + {|a-b|} \right)\\
&\ge& \frac{1}{2} \left( \left|\sqrt{2} - \frac{\tau^2}{\sqrt{2}}\right|
\right) + \frac{\tau}{2}  = \frac{1}{\sqrt{2}} - \frac{\tau^2}{2\sqrt{2}} +
\frac{\tau}{2}  \ge \frac{1}{2} + \frac{\tau}{1000}
\end{eqnarray*}
Again, the last inequality uses that $\tau \le \newa{2/5}$. This finishes the proof of Lemma~\ref{lem:non-can}.
\end{proof}


\subsection{Proof of Theorem~\ref{thm:main-kk}} \label{ssec:main-kk}

We will prove that if $w \in \R^n$ is a canonical vector such that $\Inf_1(\ell) \geq w_1^2-c_1$, then $\w^{> 2}[\ell] \geq c_2.$
For the sake of intuition, we start by providing a proof sketch for the special case that $c_1 = 0$.
At a high-level, the actual proof will be a robust version of this sketch using the notion of the critical index to make the simple arguments
for the ``$c_1=0$ case'' robust. For this case, it suffices to prove the following implication:
\begin{quote}
If $\Inf_1(\ell) = w_1^2$, then at least a constant fraction of the Fourier weight
of $\ell$ lies above level $2$.
\end{quote}
\noindent Indeed, we have the following claims:
\begin{enumerate}
\item [(1)] Let $w$ be canonical and $\Inf_1(\ell) = w_1^2.$  Then $w$ {\em equals}
$(w_1,\ldots,w_1, 0,\dots,0)$ where there are $k$ repetitions of $w_1$ and $k$ is even.  We call such a $w$  ``good''.
\item [(2)]
Let $w$ be a good vector.  Then $\ell$ has $\Theta(1)$ Fourier weight above level $2$.
\end{enumerate}
We can prove $(1)$ as follows.  Suppose that $\Inf_1(\ell)=w_1^2$.  Then, as implied by the proof of Lemma~\ref{lem:inf-ell},
\emph{every} outcome $\rho^{(2)}$ of
$(x_2,\dots,x_n)$ has $|w^{(2)} \cdot \rho^{(2)}| \geq w_1.$
Suppose, for the sake of contradiction, that some coordinate $w_j$ is neither equal to $w_1$ nor to $0$.
Let $w_k$ ($k \geq 2$) be the first such value.  By having
$\rho_2,\dots,\rho_{k-1}$
alternate between $+1$ and $-1$ we can ensure that there is an assignment of $\rho_2,\dots,\rho_{k-1}$ such that
$w_2  \rho_2 + \cdots+ w_{k-1} \rho_{k-1}$ is either $0$ (if $k$ is even) or $w_1$ (if $k$
is odd).  In the former case, by choosing the remaining $\rho$ bits
appropriately we get that there exists an assignment $\rho$ such that
$|w^{(2)} \cdot \rho^{(2)}| \leq w_k < w_1$, where the inequality
uses the fact that the $w_i$'s are non-increasing and our assumption that $w_k \ne w_1$.  In the latter case,
if $w_k$ is the last nonzero entry, for an appropriate $\rho$, we can get
$|w^{(2)} \cdot \rho^{(2)}| = w_1 - w_k < w_1.$  Otherwise, if there
are other nonzero entries beyond $w_k$ we can similarly get
$|w^{(2)} \cdot \rho^{(2)}| < w_k$.  So we have argued
that if there is any $w_k \notin \{0,w_1\}$ then it cannot be the
case that $\Inf_1(\ell)=w_1^2$, so $w$ must be of the form ($k$ copies
of $w_1$ followed by $0$'s).  If $k$ is odd, then clearly there exists a $\rho$ such that
$|w^{(2)} \cdot \rho^{(2)}| = 0$. So,  it must be the case that $k$ is even.
This proves $(1)$.
Given $(1)$ in hand, we may conclude $(2)$ using the following lemma
(Lemma~\ref{lem:good-vector}) and the observation \newa{(recalling that $w$ is canonical)}
that since $w_1\ge 0.3$
we must have $k \le 12$:
\begin{lemma}\label{lem:good-vector}
Let $\ell_k(x) =\left| \frac{(x_1 + \ldots +x_k)}{\sqrt{k}}\right|$. For $k \ge 4$ and even, $\w^{\ge 4}[\ell] \ge \frac{2^{-2k}}{k}$.
\end{lemma}
\begin{proof}
We start by observing that because $\ell_k(x)$ only takes values which are
integral multiples of $k^{-1/2}$, it must be the case that for any
character $\chi_S$, the value $\widehat{\ell_k}(S) = \mathbf{E}
[\chi_{\newa{S}}(x) \cdot \ell_k(x)]$ is a multiple of $2^{-k} \cdot
k^{-1/2}$.  Hence, any non-zero Fourier coefficient
of $\ell_k$ is at least $2^{-k} \cdot k^{-1/2}$ in magnitude.
Thus, if $\w^{\ge 4}[\ell]  \not =0$, then $\w^{\ge 4}[\ell]
\ge k^{-1} 2^{{\newa{-2k}}}$. Thus, to prove the lemma, we need to show
that $\w^{\ge 4}[\ell]  \not =0$.

Next, we observe that $\ell_k(x)$ is an even function and hence any Fourier coefficient $\hat{f}(S)=0$ if $|S|$ is odd. Thus, towards a contradiction, if we assume that $\w^{\ge 4}[\ell]  =0$, then the Fourier expansion of $\ell_k(x)$ must consist solely of a  constant term and degree $2$ terms. As the function $\ell_k(x)$ is symmetric, we may let the coefficient of any quadratic term be $\alpha$ and the constant term be $\beta$, and we have
\begin{eqnarray*}
\ell_k(x) &=& \beta + \sum_{i<j} \alpha \cdot x_i x_j = \beta + \alpha \cdot \left(\sum_{i<j}  x_i x_j\right) = \beta +\frac{ \alpha \cdot \left((\littlesum_{i=1}^k x_i)^2 - \littlesum_{i=1}^k x_i^2\right)}{2} \\ &=& \beta +\frac{ \alpha \cdot \left((\littlesum_{i=1}^k x_i)^2 -k\right)}{2} =\frac{\alpha}{2} \cdot \left(\littlesum_{i=1}^k x_i \right)^2  + \beta - \frac{\alpha k}{2}  = \gamma_1 \left(\littlesum_{i=1}^k x_i \right)^2  + \gamma_2
\end{eqnarray*}
where $\gamma_1= \alpha/2$ and $\gamma_2 = \beta - \frac{\alpha k}{2}$.
Note that since $k$ is even, there exist assignments $x \in \{-1,1\}^k$ that cause
$\sum_{i=1}^k x_i$ to take any even value in $[-k,k]$; in particular, since $k \ge 4$, the sum
$\sum_{i=1}^k x_i$ may take any of the values 0,2,4.

Now, if $\littlesum_{i=1}^k x_i=0$, then $\ell_k(x)=0$. Hence we infer that $\gamma_2=0$.  If $\littlesum_{i=1}^k x_i=2$ then $\ell_k(x) = 2/\sqrt{k}$, and if $\littlesum_{i=1}^k x_i=4$
then $\ell_k(x) = 4/\sqrt{k}$. Clearly, there is no $\gamma_1$ satisfying both $\gamma_1 \cdot 2^2 = 2/\sqrt{k}$ and $\gamma_1 \cdot 4^2 = 4/\sqrt{k}$.  This gives a contradiction.  Hence $\w^{\ge 4}[\ell] \not =0$ and the lemma is proved.
\end{proof}
\medskip

\noindent We can now proceed with the formal proof of Theorem~\ref{thm:main-kk}.
We will need several facts and intermediate lemmas.
The first few facts show some easy concentration properties for
weighted linear combinations of random signs under certain conditions on the
weights.

\newa{

\begin{claim} \label{claim:basic}
Fix $\alpha > 0.$
Let $w_1,\dots,w_n \in \R$ satisfy $|w_i| \leq \alpha$ for all $i$.
Then there exists $x^* \in \{-1,1\}^n$ such that $w \cdot x \in [0,\alpha]$
(and clearly $-x^* \in \{-1,1\}^n$ has $w \cdot (-x^*) \in [-\alpha,0]$).
\end{claim}

\begin{proof}
Construct $x' \in \{-1,1\}^n$ one bit at a time, by choosing $x'_{i+1}$
so that $\sign(w_{i+1}x'_{i+1}) =
-\sign(w_1 x'_1 + \cdots + w_i x'_i)$. The resulting vector $x'$ satisfies
$|w \cdot x'| \leq \alpha.$
\end{proof}

As a special case of this we get:

\begin{claim}\label{claim:odd}
Fix $0< \eta \leq \alpha$.
Let $w_j \in \R^+$, $j \in [2k+1]$, satisfy $w_j \in [\alpha-\eta, \alpha]$. Then, there exists $x^{\ast} = (x^{\ast}_1, \ldots, x^{\ast}_{2k+1}) \in \bits^{2k+1}$ such that:
$\littlesum_{j=1}^{2k+1}w_j x^{\ast}_j \in [0, \alpha]$.
\end{claim}

The following claim is only slightly less immediate:

\begin{claim}\label{claim:even} Fix $0< \eta \leq \alpha$.
Let $w_j \in \R^+$, $j \in [2k]$, satisfy $w_j \in [\alpha-\eta, \alpha]$. Then, there exists $x^{\ast} = (x^{\ast}_1, \ldots, x^{\ast}_{2k}) \in \bits^{2k}$ such that:
$\littlesum_{j=1}^{2k}w_j x^{\ast}_j \in [0, \eta]$.
\end{claim}
\begin{proof}
The vector $u \in \R^k$ defined by $u_j = w_{2j}-w_{2j-1}$ has $|u_j| \leq
\eta$ for all $j \in [k]$.  It is clear that the set of values
$\{w \cdot x\}_{x \in \{-1,1\}^{2k}}$ is contained in $\{u \cdot x\}_{x
\in \{-1,1\}^k}.$  The claim follows by applying Claim~\ref{claim:basic}
to $u.$
\ignore{
Note that, by symmetry, it suffices to show that there exists $x^{\ast}$ with $\littlesum_{j=1}^{2k}w_j x^{\ast}_j \in [-\eta, \eta]$.  The proof is by induction on $k$. The claim is obvious for $k=1$. Assume it is true for $k=j$. Without loss of generality, we can assume we have $x'_1, \ldots, x'_{2k}$, with $0 \leq \littlesum_{i=1}^{2j} w_i x_i \leq \eta$. It is easy to see that there exist $x^{\ast}_{2j+1} $ and $x^{\ast}_{2j+2}$ with the property that $x^{\ast}_{2j+1} w_{2j+1} + x^{\ast}_{2j+2} w_{2j+2} \in [-\eta, 0]$. This choice completes the induction step.
}
\end{proof}
}

We will also need the following corollary of the Berry-Ess{\'e}en theorem
\newa{
(more precisely, it follows from Fact~\ref{fact:be} together with the
fact that the pdf of a standard Gaussian has value at least
$0.2$ everywhere on $[-1,1]$):}

\begin{fact}\label{fac:regular} Fix $0<\tau<1/\newa{15}$.
Let $w \in \mathbb{R}^n$ be $\tau$-regular with $\| w \|_2 \leq 1$.
Then, $\Pr_{x} [0 \leq w \cdot x  \leq \newa{15} \tau] \geq \tau$ and
$\Pr_{x}  [\newa{-15} \tau \leq w \cdot x  \leq 0] \geq \tau$.
\end{fact}

We are now ready to prove the following lemma which establishes a concentration statement for linear forms with a given maximum coefficient:

\begin{lemma}\label{lem:2}
Let $w \in \R^n$ be proper with $\| w \|_2 \leq1$ and let  $\delta \eqdef w_1 >0$. There exists $\kappa = \kappa(\delta)$ such that
$\Pr_{x} [0 \leq w \cdot x \leq \delta] \geq \kappa$.
\end{lemma}
\begin{proof}
We choose a sufficiently small $\tau>0$, where $\tau=\tau(\delta) \ll \delta$,
and consider the $\tau$-critical index $K\newa{=c(w,\tau)}$ of $w$.
Fix $K_0 = \Theta (1/\tau^2) \cdot \log (1/\delta^2)$
and consider the following two cases:

\smallskip

\noindent {\bf [Case 1: $K \leq K_0.$]}
In this case, we partition $[n]$ into the head $H = [K-1]$ and the tail $T = [n] \setminus H.$
Then, an application of Claims~\ref{claim:even} and~\ref{claim:odd} for $\eta = \alpha = \delta$ gives us that
$$\Pr_{x_H} \left[ w_H \cdot x_H  \in [0, \delta] \right] \ge 2^{-K} \geq 2^{-K_0}.$$
An application of Fact~\ref{fac:regular} for the $\tau$-regular tail gives that
$$\Pr_{x_T} \left[ w_T \cdot x_T  \in [-\newa{15}\tau, 0] \right] \ge \tau$$
and combining the above inequalities using independence yields
$$\Pr_x \left[ w \cdot x \in [-\newa{15}\tau, \delta] \right]
\ge 2^{-K_0} \cdot \tau.$$
Now note that for any choice of $\tau \le \delta/\newa{15}$, the above clearly 
implies
$$\Pr_x \left[ w \cdot x \in [-\delta, \delta] \right] \ge 2^{-K_0} \cdot \tau$$
and by symmetry we conclude that
$$\Pr_x \left[ w \cdot x \in [0, \delta] \right] \ge 2^{-K_0-1} \cdot \tau$$
yielding the lemma for $\kappa_1 = 2^{-K_0-1} \cdot \tau.$
\smallskip

\noindent {\bf [Case 2: $L > K_0.$]}
In this case, we partition $[n]$ into $H = [K_0-1]$ and the tail $T = [n] \setminus H.$
We similarly have that $$\Pr_{x_H} \left[ w_H \cdot x_H \in [0, \delta] \right] \ge 2^{-K_0}.$$
Now recall that the tail weight decreases geometrically up to the critical index; in particular, Fact~\ref{fact:small-tail} gives that $\| w_T \|_2 \le \delta^2$. Then, for a sufficiently small $\delta$,  the Hoeffding bound gives $$\Pr_{x_T} \left[ w_T \cdot x_T \in [-\delta, 0] \right] \geq 1/4.$$  Combining these inequalities we thus get that
$$\Pr_x \left[ w\cdot x \in [-\delta, \delta] \right] \ge 2^{-K_0-\newa{2}}.$$
By symmetry, we get the desired inequality for $\kappa_2 = 2^{-K_0-\newa{3}}$.

\smallskip

\noindent The proof follows by selecting $\kappa = \min \{\kappa_1, \kappa_2 \} = \kappa_1$ for any choice of $\tau \leq  \delta/\newa{15}.$
\end{proof}

\newb{Note the difference between the conditions of Corollary~\ref{cor:1},
stated below, and Lemma~\ref{lem:2} stated above:  while Lemma~\ref{lem:2}
requires that $\delta = w_1$, Corollary~\ref{cor:1} holds
for any $\delta > 0.$}

\begin{corollary}\label{cor:1}
For any $\delta>0$, there is a value $\kappa=\kappa(\delta)>0$
such that for any $w \in \mathbb{R}^n$  with $\|w\|_2 \leq1$ and $\| w \|_{\infty} \leq \delta$,
$\Pr_{x} [0\le w \cdot x \le \delta] \ge \kappa$ and $\Pr_{x} [-\delta \le w \cdot x \le 0] \ge \kappa.$
\end{corollary}
\begin{proof}
We start by considering the case when $\Vert w \Vert_2 \le \delta/\newa{100}$.
In this case, by Theorem~\ref{thm:chb}, we certainly get that
$\Pr_x [|w \cdot x| \le \delta] \ge 99/100$.  Hence, by symmetry,
$\Pr_x [-\delta \le w  \cdot x \le 0] \ge 99/200$ and
$\Pr_x [0 \le  w \cdot x \le \delta] \ge 99/200$.

Next, we consider the case when $\Vert w \Vert_2 > \delta/\newa{1500}$.
In this case, if $w_1 > \delta^2/\newa{1500}$,
then we apply Lemma~\ref{lem:2}, to get that $\Pr_x [0 \le w
\cdot x \le \delta^2/\newa{1500}] \ge \kappa_1$ and (by symmetry)
$\Pr_x [-\delta^2\newa{/1500} \le w \cdot x \le 0] \ge \kappa_1$
where $\kappa_1$ is a positive constant dependent only on $\delta$.

The only remaining case is when $w_1 \le \delta^2/\newa{1500}$.
In this case, the vector $w$ is $ \delta/\newa{15}$-regular.
Now, we can apply Fact~\ref{fac:regular} to get that $\Pr_x[0 \le w
\cdot x \le \delta] \ge \delta/\newa{15}$ and $\Pr_x[-\delta  \le w
\cdot x \le 0] \ge \delta/\newa{15}$.
By taking $\kappa = \min \{\delta/\newa{15}, \kappa_1, 99/200 \}$,
the proof is completed.
\end{proof}

Using the above corollary, we show the following lemma:

\begin{lemma}\label{lem:3}
Let $\alpha, \eta,\xi \in \mathbb{R}^+$ with $w \in \R^n$ be such that
$\xi \le \| w \|_2 \le 1$, $\alpha > 2 \eta$, and $\| w \|_{\infty}
\le \alpha - \eta$. Then, there are positive constants $\kappa = \kappa(\alpha,\eta,\xi)$ and $\gamma = \gamma (\alpha, \eta,\xi)$ such that
$$
\Pr_x [-2 \alpha + 2\eta \le  w \cdot x \le -\gamma] \ge \kappa.
$$
\end{lemma}
\begin{proof}
We choose a sufficiently small $\zeta>0$ and consider two cases.

\smallskip

\noindent \noindent {\bf [Case 1: $w$ is $\zeta$-regular.]} In this case,
Theorem~\ref{thm:be} gives us \newa{(similar to Fact~\ref{fac:regular})
that for $\zeta \leq 1/20$, we have}
$$
\Pr_{x} [ -20 \zeta \cdot \Vert w \Vert \le w \cdot x  \le -\zeta \cdot \Vert w \Vert] \ge \zeta.
$$

\smallskip

\noindent \noindent {\bf [Case 2: $w$ is {\em not} $\zeta$-regular.]}
We assume without loss of generality that $w_1 = \|w\|_\infty.$
In this case, it follows by definition that $w_1 \ge \zeta \cdot \xi$, hence  $w_1 \in [\zeta \cdot \xi, \alpha-\eta]$.
Since $|w_{j}| \le \alpha - \eta$ for all
$j \geq 2$, Corollary~\ref{cor:1} says that
(recall that $w^{(1)} = (w_2,\dots,w_n)$)
$$
\Pr_{x^{(1)}} \left[ -\alpha + \eta \le w^{(1)} \cdot x^{(1)} \le  0 \right] \ge c(\alpha,\eta).
$$
By independence we thus get
$$
\Pr_{x} \left[ -2\alpha + 2\eta \le w\cdot x \le -\zeta \cdot \xi  \right] \ge c(\alpha,\eta)/2.
$$

Combining \textbf{Case 1} and \textbf{Case 2} and using $1 \ge \Vert w \Vert \ge \xi$, we get
$$
\Pr_{x} \left[\min\{ -2\alpha + 2\eta,-20 \zeta   \} \le w\cdot x \le -\zeta \cdot \xi  \right] \ge \min \{c(\alpha, \eta)/2 , \zeta \}.
$$
We now choose $\zeta>0$ so that $20 \zeta \le  \alpha -\eta$.
Finally, we set $\gamma = \zeta \cdot \xi$ and  $\kappa =  \min \{c(\alpha, \eta)/2 , \zeta \}$ and get the claimed result.
\end{proof}

The next lemma is a robust version of Lemma~\ref{lem:good-vector}.  It says that if a vector $w$
of length $n$ is very close to having its first $2k$ entries each being $\alpha$ and
its remaining entries all 0, then $\ell(x)=|w \cdot x|$ must have nonnegligible
Fourier mass at levels 4 and above.

\begin{lemma}\label{lem:fourier-mass}
Let $\alpha>0$, $w \in \mathbb{S}^{n-1}$ and $k \in \mathbb{N}$, $k>1$. Then there are sufficiently small positive constants $\eta = \eta(k)$, $\tau = \tau (k)$  with the following property : If for every $i \in [2k]$, we have
$w_i \in [\alpha - \eta, \alpha]$ and $\littlesum_{j>2k}^n (w_j)^2 \le \tau^2$, then the map $\ell : x \mapsto |w \cdot x|$  satisfies $\w^{\ge 4} [\ell]
\ge \gamma$ for some $\gamma = \gamma(k) >0$.
\end{lemma}
\begin{proof}
Consider the vector $w' = (\underbrace{\alpha, \ldots, \alpha}_{2k}, 0, \ldots, 0)$ and the map $\ell' : x \mapsto |w' \cdot x|$. We have
$$
\ell'(x) = \alpha \cdot \sqrt{k} \cdot \left| \frac{x_1 + \ldots +x_k}{\sqrt{k}} \right|
$$
By applying Lemma~\ref{lem:good-vector}, we get $\w^{\ge 4} [ \ell'] \ge
\newa{\alpha^2} \cdot 2^{-2k}$. Note that if $\eta$ and $\tau$ are sufficiently small, then clearly $\alpha \ge \frac{1}{2\sqrt{k}}$. This implies $ \w^{\ge 4} [ \ell']  \ge \frac{2^{-2k}}{\newa{4k}}$.

We now observe that
\begin{eqnarray*}
|\ell(x) - \ell'(x)| &=& \left| \ | \littlesum_{i=1}^n w_i \cdot x_i|  - |\littlesum_{i=1}^n w'_i \cdot x_i| \ \right|\\ 
&\le&  \left| \ | \littlesum_{i=1}^k w_i \cdot x_i|  - |\littlesum_{i=1}^k w'_i \cdot x_i| \ \right| + \left|\littlesum_{j=k+1}^n w_i x_i \right|
\end{eqnarray*}
Let us use $h_1(x) = | \littlesum_{i=1}^k w_i \cdot x_i| $, $h_2(x) =  |\littlesum_{i=1}^k w'_i \cdot x_i| $ and $h_3(x) = |\littlesum_{j=k+1}^n w_i x_i |$.
Then we may rewrite the above as $|\ell(x) - \ell'(x) | \le |h_1(x) - h_2(x)|
+ h_3(x)$.  This implies that
$|\ell(x) - \ell'(x)|^2 \le 2(h_1(x) - h_2(x))^2 + 2 (h_3(x))^2$.  This in
turn yields
$$
\mathbf{E} [(\ell(x) - \ell'(x) )^2] \le 2 \mathbf{E} [(h_1(x) - h_2(x))^2] + 2 \mathbf{E} [(h_3(x))^2].
$$
Note that $\mathbf{E} [(h_3(x))^2] = \littlesum_{j=k+1}^n w_j^2 \le \tau^2$. Next, observe that
$$
|h_1(x) - h_2(x)| = \left|  | \littlesum_{i=1}^k w_i \cdot x_i|  - |\littlesum_{i=1}^k w'_i \cdot x_i|  \right| \le \left|\littlesum_{i=1}^k (w_i - w'_i) \cdot x_i\right|
$$
Hence, we get that $\mathbf{E} [(h_1(x) - h_2(x))^2] \le \mathbf{E} [(\littlesum_{i=1}^k (w_i - w'_i) \cdot x_i)^2] \le \littlesum_{i=1}^k \eta^2 = k \eta^2$.

Combining these bounds,
we get that $\mathbf{E} [(\ell(x) - \ell'(x) )^2]  \le 2 (k\eta^2 + \tau^2)$.
Hence, we have that
$$\w^{\ge 4} [\ell] \ge \w^{\ge 4} [\ell'] -
\mathbf{E} [(\ell(x) - \ell'(x) )^2]  \ge \frac{2^{-2k}}{\newa{4}k}
-2 (k\eta^2 + \tau^2).$$
We may choose $\eta$ and $\tau$ small enough so that $\frac{2^{-2k}}{\newa{4}k} -2 (k\eta^2 + \tau^2) \ge \frac{2^{-2k}}{\newa{8}k}$, and the proof is finished.
\end{proof}

Given the above lemmas, the proof of Theorem~\ref{thm:main-kk} proceeds as follows:
Let $w_1 = \alpha$ and $\eta= \eta(\alpha)>0$ be a sufficiently small constant. Let $L$ be the first index such that $w_{L} \le \alpha -\eta$.
\newa{Recalling that $w$ is canonical,} since $w_1>0.3$ and $\|w\|=1$, it is clear that $L \leq \newa{1/0.09 < 12}.$
We now consider two cases  :

\smallskip

\noindent {\bf [Case I: $L$ is even]} Then by Claim~\ref{claim:even}, there is a choice of $x_2, \ldots, x_{L-1}$, such that $\littlesum_{k=2}^{L-1} w_k x_k \in [-\eta, 0].$  Using Corollary~\ref{cor:1} and noting that $w_{L} \le \alpha -\eta$, there is some $\kappa = \kappa(\alpha,\eta)$ such that $\Pr_{x^{(L-1)}} [0 \le w^{(L-1)}x^{(L-1)} \le \alpha - \eta] \ge \kappa$. By independence, we thus get
\ignore{
$$
\Pr_{x} [- \eta \le w^{(1)} \cdot x^{(1)} \le \alpha - \eta] \ge \kappa \cdot 2^{-L}.
$$
That is, for $\eta$ sufficiently small, we get
}
\begin{equation} \label{eqn:ena}
\Pr_{x} [ -  \eta \le w^{(1)} \cdot x^{(1)} \le \alpha - \eta] \ge \kappa \cdot 2^{-L}.
\end{equation}
Note that (\ref{eqn:ena}) implies (by definition) that
$\Inf_1(\ell) \leq w_1^2-c_1$, for an appropriate constant
$c_1 = c_1(\kappa, L\newa{,\eta})>0.$

\smallskip

\noindent {\bf [Case II: $L$ is odd]}
Let us choose a sufficiently small $\xi >0$.
If $\| w^{(L-1)} \|_2 > \xi$, then observe that from
Claim~\newa{\ref{claim:even} (applied to the weights $w_1,\dots,w_{L-1}$)}
there is a choice of $x_2, \ldots, x_{L-1}$
satisfying $\littlesum_{k=2}^{L-1} w_k x_k \in [\alpha-\eta, \alpha]$,
i.e.
$$
\Pr \left[\alpha- \eta \le \littlesum_{k=2}^{L-1} w_k x_k \le \alpha\right]  \ge 2^{-L}.
$$
Combining this with Lemma~\ref{lem:3} applied to $w^{{(L-1)}}$, we get that
\begin{equation} \label{eqn:dyo}
\Pr_{x^{(1)}} [-\alpha + \eta \le w^{(1)} \cdot x^{(1)} \le \alpha - \gamma(\newa{\alpha,\eta,}\xi) ] \ge 2^{-L} \cdot \kappa.
\end{equation}
Exactly as before, (\ref{eqn:dyo}) implies (by definition) that $\Inf_1(\ell) \leq w_1^2-c_1$, for an appropriate constant $c_1>0.$

Now consider the only remaining case which is that
$\| w^{(L-1)} \|_2 \le \xi$.
Recall that $1 < L< 12$ and $L$ is odd; we first claim that
that $L>3$.  Indeed, this must be the case because $L = 3$ contradicts
(for $\xi$ and $\eta$ sufficiently small)
the assumption $\tau \ge \newa{2/5}$ \newa{(recall that $w$ is canonical)}.
Now, since $\ell \le 11$ and $\eta$ and $\xi$ are sufficiently small,
by applying
 Lemma~\ref{lem:fourier-mass}), we get that  $\ell$ has a constant fraction of its Fourier mass above level $2$, completing the proof.
This finishes the proof of Theorem~\ref{thm:main-kk}.


\fi



\section{Proof of Theorem~\ref{thm:w1} using Theorem~\ref{thm:robust-kk}} \label{sec:w1-k}

We first observe that it suffices to prove the theorem for balanced LTFs,
i.e. LTFs $f:\bn \to \bits$ with $\newa{\widehat{f}(\emptyset)=}\E[f] = 0$.
(Note that any balanced LTF can be represented with a threshold of $0$,
i.e. $f(x) = \sgn (w \cdot x)$ for some $w \in \R^n$.) \newa{
\begin{fact}\label{fac:balanced-w1}
Let $f:\bn \to \bits$ be an  $n$-variable
LTF. Then there is a balanced $(n+1)$-variable LTF $g : \{-1,1\}^{n+1} \rightarrow \{-1,1\}$ such that $\w^{\leq 1} [f]  =\w^{\leq 1} [g].$ \end{fact}
 \begin{proof} Let $f(x) = \sgn(w_0 + \littlesum_{i=1}^n w_ix_i)$
and note that we may assume that $w_0 \neq w \cdot x$ for all
$x \in \{-1,1\}^n.$
Consider the $(n+1)$-variable balanced LTF $g:(x, y) \to \bits$,
where $y \in \bits$, defined by
$g(x, y) = \sgn(w_0 y + \littlesum_{i=1}^n w_ix_i)$. Then it is easy to see that $\wh{g}(y) = \E[f]$ and $\wh{g}(i) = \wh{f}(i)$
for all $i \in [n]$. Therefore, $\w^{\leq 1} [f]  = \w^{1} [g] = \w^{\leq 1} [g].$
\end{proof}}
\medskip

Let $f = \sgn(w \cdot x)$ be an LTF. We may assume that $w$ is a proper unit
vector, i.e. that $\| w\|_2=1$ and $w_i \geq w_{i+1} >0$ for $i \in [n-1].$
We can also assume that $w \cdot x \neq 0$ for all $x \in \bn.$
We distinguish two cases: If $w$ is ``far'' from $w^{\ast}$ (i.e. the
worst-case vector for the Khintchine inequality),
the desired statement follows immediately from our robust inequality
(Theorem~\ref{thm:robust-kk}).
For the complementary case, we use a separate argument that exploits the
structure of $w$.  More formally, we have the following two cases:

\medskip

Let $\tau>0$ be a sufficiently small universal constant, to be specified.

\medskip

\noindent {\bf [Case I: $\|w - w^{\ast}\|_2 \geq \tau$]}. In this case, Proposition~\ref{prop:gl} and
Theorem~\ref{thm:robust-kk} give us
$$\w^1[f] \geq \left( \kk(w) \right)^2 \geq (1/\sqrt{2}+ c\tau )^2 \geq 1/2 + \sqrt{2} c  \tau$$
which completes the proof \newa{of Theorem~\ref{thm:w1}} for Case I.

\medskip

\noindent {\bf [Case II: $\|w - w^{\ast}\|_2 \leq \tau$]}.
In this case
\newa{the idea is to consider the restrictions of $f$ obtained by fixing the
variables $x_1, x_2$ and argue based on their bias.}
\ifnum\short=1
\newa{See the full version for details.}
\fi
\ifnum\short=0
\newa{Recall that} for a vector $y = (y_1, \ldots, y_n) \in \R^n$ and $i \in [n]$
we denote $y^{(i)} = (y_{i+1}, \ldots, y_n)$.
We consider the restrictions $f_{ij}: \bits^{n-2} \to \bits$ defined by
$$f_{ij}(y) =  \sgn(w_1 \cdot (-1)^{i} + w_2 \cdot (-1)^{j} + w^{(2)} \cdot y).$$
We fix $\lambda=3/4$ and consider the following two subcases:

\begin{enumerate}
\item[(a)] ($\E_y [f_{01}(y)] \leq \lambda$) In this case
the function $f_{01}$ is not very positively biased;
we show that the Cauchy-Schwarz inequality is not tight.
In particular, the degree-$1$ Fourier vector
$\newa{(\widehat{f}(i))_{i=1,\dots,n}}$ of $f(x)  = \sgn(w\cdot x)$
and the corresponding weight-vector $w$ form an angle bounded away from zero:

\begin{lemma} \label{lem:w1-close}
There are universal constants $\tau, \kappa = \kappa(\tau) >0$ such that the following holds: Let $w \in \mathbb{R}^n$ be any proper unit vector \newa{such that $\| w - w^{\ast} \|_2 \leq \tau$ and $\E_y [f_{01}(y)] \leq \lambda$ where $f(x)=\sign(w \cdot x)$.}  Then we have
$$\w^1[f] \geq  (1+\kappa) \cdot \left( \kk(w) \right)^2.$$
\end{lemma}

\begin{proof}
\newa{Note that since $w_1 \geq w_2$ the function $f_{01}(y)$
is an LTF of the form $\sign(w^{(2)} \cdot y^{(2)} + \theta)$
with $\theta \geq 0$, and hence $\E[f_{01}] \geq 0.$}
To deal with this case we recall the following simple fact:

\begin{fact}[Lemma~2.4 in \cite{OS11:chow}] \label{fact:not-biased-deg1}
Let $f:\bn \to \bits$ be an LTF with $1- |\E[f]| = p$. Then
$\w^1[f] \geq p^2/2. $
\end{fact}

An application of Fact~\ref{fact:not-biased-deg1} for $f_{01}$ gives $$\w^1[f_{01}]   \geq 1/(32).$$
Note that by symmetry we also have that $\E_y [f_{10}(y)] \geq -\lambda$ and therefore
$$\w^1[f_{10}]   \geq 1/(32).$$
Fix $k \in \{3,\dots,n\}$. We have that
\begin{eqnarray*}
\wh{f}(k) &=& \Inf_k(f)\\
& = & (1/4) \cdot \littlesum_{i, j \in \newa{\{0,1\}}}
\Inf_{k-2} (f_{ij}) \geq (1/4) \cdot \left( \wh{f_{01}}(k-2)+ \wh{f_{10}}(k-2)
\right).
\end{eqnarray*}
\newa{Since the sign of $\widehat{f_{01}}(k-2)$
agrees with the sign of $\widehat{f_{10}}(k-2)$ for
all $k \in \{3,\dots,n\}$,}
we get that
$$
\littlesum_{k=3}^n \wh{f}(k)^2  \geq (1/16) \cdot \left( \w^1[f_{01}]  + \w^1[f_{10}] \right) \geq 1/(256).
$$

Recall that by assumption of the lemma it holds $\|w^{(2)}\|_2 = \sqrt{\littlesum_{i=3}^n w_i^2} \leq \tau$ and Parseval's identity implies that
$\littlesum_{i=1}^n \wh{f}^2(i)^2 \leq 1$. We can therefore now write
\begin{eqnarray*}
\kk(w) = \littlesum_{i=1}^n \wh{f}(i) w_i
&\leq& \sqrt{w_1^2+w_2^2} \cdot  \sqrt{\wh{f}^2(1)+\wh{f}^2(2)} + \|w^{(2)}||_2 \cdot \sqrt{\littlesum_{k=3}^n \wh{f}(k)^2} \\
&\leq& \sqrt{\w^1[f] - 1/(256)} + \tau
\end{eqnarray*}
where the first inequality follows by two applications of Cauchy-Schwarz and the second follows by our assumptions. By squaring and expanding, assuming that
$\tau>0$ is sufficiently small, we obtain
\begin{eqnarray*}
\left( \kk(w) \right)^2  &\leq&  \w^1[f] - 1/300 \\
&\leq&  \w^1[f] - (1/300)  \w^1[f]  = (299/300) \cdot  \w^1[f]
\end{eqnarray*}
where the second inequality follows from the fact that $\w^1[f] \leq 1$.  This
proves Lemma~\ref{lem:w1-close}.
\end{proof}

Theorem~\ref{thm:w1} follows easily from Lemma~\ref{lem:w1-close} in this subcase using the
``basic'' Khintchine inequality with optimal constant, $(\kk(w))^2 \geq 1/2.$
\newa{We turn now to the remaining subcase:}

\item[(b)] ($\E_y [f_{01}(y)] > \lambda = 3/4$)  In this case, we show that
\newa{the value $\widehat{f}(1)$ is so large that it alone causes $\w^{\leq 1}[f]$ to be
significantly larger than $1/2.$
Since $\E_y[f_{01}(y)] > 3/4$ it must certainly also be the case
that $\E_y[f_{00}(y)] > 3/4$, and by symmetry
$\E_y[f_{10}(y)] < -3/4$ and $\E_y[f_{11}(y) < -3/4$. Consequently we have
$\widehat{f}(1) = \E_x[f(x)x_1] > 3/4,$
and so $\w^{\leq 1}[f] \geq \widehat{f}(1)^2 \geq 9/16.$

This concludes the proof of Theorem~\ref{thm:w1}.

}

\end{enumerate}

\fi



\section{Alternate proof of Theorem~\ref{thm:w1}} \label{sec:w1-os}

Recall that it suffices to prove the theorem for balanced LTFs.
The idea of the second proof is to perturb the original
halfspace slightly so that the perturbed halfspace is defined by
a sufficiently anti-concentrated linear form $w' \cdot x$.
If the perturbed halfspace is regular, one can show
that its degree-$1$ Fourier weight is close to $2/\pi$. Otherwise, there exists a large weight, hence an influential variable $x_1$ (say).
We are then able to show a non-trivial upper bound
on the influence of $x_1$ on the function $\ell(x) = |w' \cdot x|$.

We require the following terminology:

\begin{definition}
The \emph{ (relative) Hamming distance}
between two Boolean functions $f, g : \bn \to \bits$ is defined as follows: $\dist(f,g) \eqdef \Pr_{x}[f(x) \neq g(x)].$
If $\dist(f,g) \leq \eps$ we say that $f$ and $g$ are \emph{$\eps$-close}.
\end{definition}

\begin{definition}
Let $f:\bn \to \bits$ be an LTF, $f(x) = \sgn(w_0+\littlesum_{i=1}^n w_i x_i)$, where the weights are scaled
so that $\littlesum_{i=0}^n w_i^2=1$. Given a particular input $x \in \bn$ we define $\margarine (f,x) = |w_0+\littlesum_{i=1}^n w_i x_i|$.
\end{definition}

We start by recalling the following result from~\cite{OS11:chow} which essentially says that any LTF is extremely close to another LTF for which almost all points have large margin:

\begin{theorem} \label{thm:os-perturb}[Theorem 6.1 in~\cite{OS11:chow}]
Let $f:\bn \to \bits$ be any LTF and let $0<\tau<1/2$. Then there is an LTF $f' : \bn \to \bits$ with $\dist(f,f') \leq \eta(\tau)$ satisfying
$\Pr_x[\margarine(f', x) \leq \kappa(\tau)] \leq \tau$, where $\kappa(\tau) = 2^{-O(\log^3(1/\tau)/\tau^2)}$ and $\eta(\tau) = 2^{-1/\kappa(\tau)}.$
\end{theorem}

Let $0<\tau<\delta$ be sufficiently small universal constants
(to be chosen later).
Given any balanced LTF $f(x)=\sign(w\cdot x)$, we consider the LTF
$f' = \sgn(w' \cdot x)$, $\|w'\|_2=1$,
obtained from Theorem~\ref{thm:os-perturb},
so $\dist(f, f') \leq \eta(\tau)$ and $\Pr_x[ |w' \cdot x | \leq \kappa(\tau)] \leq \tau.$
We will exploit the anti-concentration of $w' \cdot x$ to establish the theorem for $f'$. We will then use the fact that $f$ and $f'$ are close in Hamming distance to complete the theorem.

We apply Fact~\ref{fact:close} for the degree-$1$ Fourier vectors of $f$ and $f'$, i.e. $a_i = \wh{f}(i)$ and $b_i = \wh{f'}(i)$, $i \in [n]$.
Note that Parseval's identity gives that $\littlesum_{i=1}^n (\wh{f}(i))^2 \leq 1$ and $\littlesum_{i=1}^n (\wh{f'}(i))^2 \leq 1$. Moreover, Plancherel's identity implies
that $$\littlesum_{i=1}^n (\wh{f}(i) - \wh{f'}(i))^2 \leq \littlesum_{S \subseteq [n]} (\wh{f}(S) - \wh{f'}(S))^2 = \E_x [(f(x)-f'(x))^2] = 4\dist(f, f') \leq  4\eta.$$ Therefore,
\begin{equation} \label{eqn:close}
\left| \w^1[f] - \w^1[f'] \right| \leq 4\sqrt{\eta}.
\end{equation}
Therefore, Fact~\ref{fact:close} gives that

The above equation implies that if we show the theorem for $f'$ we are done as long as $\eta$ is sufficiently small. We can guarantee this by making $\tau$
sufficiently small. To show the theorem for $f'$, we consider two possibilities depending on whether the vector $w'$ defining $f'$ is $\delta$-regular
\newa{(where $\delta$ will be determined later).}


\medskip

\noindent {\bf [Case I: $w'$ is $\delta$-regular]}
In this case, we use  the following result from~\cite{MORS:10}:

\begin{theorem}[Theorem~48 in~\cite{MORS:10}] \label{thm:reg-w1}
Let $\delta>0$ be a sufficiently small universal constant  and $f$ be a $\delta$-regular LTF.  Then
$|\w^1[f] - W(\E[f]) | \leq \delta^{1/6}.$
\end{theorem}
\newa{We give a full description of the $W(\cdot)$ function in Section~\ref{ssec:gaussian-basics}; here we only will use the fact that
$W(0) = 2/\pi$.}  Theorem~\ref{thm:reg-w1} thus gives that $ \w^1[f']  \geq  \frac{2}{\pi}  - \delta^{1/6}$
and by (\ref{eqn:close}) we obtain
\begin{equation} \label{eqn:reg-bound}
 \w^1[f]  \geq  \frac{2}{\pi}  - \delta^{1/6} - 4 \sqrt{\eta}.
\end{equation}
This quantity can be made arbitrarily close to $2/\pi$ by selecting $\delta, \tau$ to be small enough constants.

\medskip

\noindent {\bf [Case II: $w'$ is {\em not} $\delta$-regular]}
In this case, we have that $|w'_1| = \max_i |w'_i| >\delta$.  Let us assume without loss of generality that $w'_1>0$. (The other case is entirely similar.)
By Proposition~\ref{prop:gl} and Fact~\ref{fac:even} we have
$$
 \w^1[f']  \geq \left(\kk(w')\right)^2 =  1 - \Var (\ell) \geq 1 - (1/2)
\cdot \Inf(\ell),
$$
where $\ell(x)=|w' \cdot x|.$
Lemma~\ref{lem:inf-ell} already implies that $\Var[\ell] \leq 1/2$, but we are able to prove a better upper bound in this case.
To prove a better upper bound on the variance, we exploit  that $w'_1>\delta$ to upper bound $\Inf_1(\ell)$ by a quantity strictly smaller than $(w'_1)^2$.
For this, we recall the following result from~\cite{MORS:10}:

\begin{theorem}[Theorem~39 in~\cite{MORS:10}] \label{thm:large-weight}
Let $f(x) = \sgn(\sum_{i=1}^n w_i x_i - w_0)$ be an LTF such that $\sum_{i} w_i^2 = 1$ and
$\delta \eqdef |w_1|$
$ \geq $
$|w_i|$ for all $i \in [n].$ Let $0 \leq
\eps \leq 1$ be such that $|\E[f]| \leq 1 - \eps$. Then
$|\hat{f}(1)| = \Omega(\delta \eps^6 \log(1/\eps)).$
\end{theorem}

\noindent We can now state and prove our main lemma for this case:

\begin{lemma}\label{lem:inf1-ell}
In the context of Case II, we have $ \Inf_1(\ell)  \leq  (w'_1)^2 - 2 (\wh{f'}(1) - 2 \tau) \kappa(\tau) w'_1+(\wh{f'}(1) - 2\tau)\kappa(\tau)^2.$
\end{lemma}

\begin{proof}
Since $w'_1>\delta$ and $\littlesum_{i=1}^n (w'_i)^2=1$,
an application of Theorem~\ref{thm:large-weight} gives $\Inf_1(f') = \wh{f'}(1) > c_1 \cdot w'_1$, where $c_1$ is a universal constant.

To analyze the desired quantity, we partition the hypercube $\bn$ into
pairs $(x^+ , x^-)$
that differ only in the fist coordinate with $x^{+}_1 = 1$ and $x^{-}_1 = -1.$ That is $x^+ = (1, y)$ and $x^- = (-1, y)$ with $y \in \bits^{n-1}.$
We say that such a pair is ``good"  if both the following conditions hold:
(1) the corresponding hypercube edge is bi-chromatic (i.e.
$f'(x^{+})=1$ and $f'(x^{-})=-1$)\footnote{This
is the only possibility since
$w'_1>0$, hence $f'$ is monotone nondecreasing in $x_1$.},
and
(2) $\min \{ |w' \cdot x^+| , |w' \cdot x^-| \} \geq \kappa(\tau)$.
It is easy to see that the fraction of pairs that are ``good''
is at least $\wh{f'}(1) - 2\tau$, i.e.
$\Pr_{y \in \bits^{n-1}} [\mathcal{G}] \geq \wh{f'}(1) - 2 \tau$,
where ${\cal G}$ is the event $\mathcal{G} = \{ y\in \bits^{n-1} ~\mid~ $
the pair $(1, y),(-1, y)$ is good$\}.$ Indeed,
the probability that the edge $(1, y)$, $(-1, y)$ is monochromatic is $1-\Inf_1(f') = 1-\wh{f'}(1)$ and the probability that either $|w' \cdot x^+| \leq \kappa(\tau)$ or
 $|w' \cdot x^-| \leq \kappa(\tau)$ is at most $\tau$, hence the claim follows by a union bound.

 Now if $y \in \bits^{n-1}$ is such that the corresponding pair $x^+ = (1, y)$ and $x^- = (-1, y)$ is good,
 we have that $|w' \cdot x^+| = w'_1 + c' \geq \kappa(\tau)$ and $|w' \cdot x^-| =  w'_1 - c' \geq \kappa(\tau)$, where $c' = (w'_2, \ldots, w'_n) \cdot y.$
 From this we deduce that $|c'| \leq |w'_1 - \kappa(\tau)| \leq |w'_1|$,
\newa{where the second inequality holds for a sufficiently small choice
of $\tau$.}
Hence, the analysis of Lemma~\ref{lem:inf-ell} yields that in this case
$\Var[\ell(x_1, y)]  = c'^2 \leq \left( w'_1 - \kappa(\tau) \right)^2.$ In all other cases, Lemma~\ref{lem:inf-ell} yields the upper bound
$\Var[\ell(x_1, y)]  = c'^2 \leq (w'_1)^2.$ We can thus bound from above the desired influence as follows:
\begin{eqnarray*}
\Inf_1(\ell) &= & \E_{y \in \bits^{n-1}} [\Var[\ell(x_1, y)]] \\
                   &\leq& (\wh{f'}(1) - 2\tau)  \cdot \left( w'_1 - \kappa(\tau) \right)^2 + (1 - \wh{f'}(1) + 2\tau)(w'_1)^2 \\
                   &\leq& (w'_1)^2 - 2 (\wh{f'}(1) - 2\tau)\kappa(\tau)w'_1+ (\wh{f'}(1) - 2\tau)\kappa(\tau)^2 .
\end{eqnarray*}
This completes the proof.
\end{proof}

Combining Lemma~\ref{lem:inf1-ell} with our earlier arguments, we obtain
$$
\w^1[f'] \geq \frac{1}{2} +  (\wh{f'}(1) - 2\tau) \kappa(\tau) w'_1 -\frac{ (\wh{f'}(1) - 2\tau)\kappa(\tau)^2}{2}
$$
and using (\ref{eqn:close}) we conclude
\begin{equation} \label{eqn:nonreg-bound}
\w^1[f] \geq  \frac{1}{2} +  (\wh{f'}(1) - 2\tau) \kappa(\tau) w'_1- \frac{ (\wh{f'}(1) - 2\tau)\kappa(\tau)^2}{2}-4\sqrt{\eta(\tau)}.
\end{equation}

At this point it is straightforward to complete the proof of Theorem~\ref{thm:w1}. Indeed,  we select $\delta>0$ to be a sufficiently small constant and $\tau \eqdef c_1 \cdot \delta /4 \ll \delta$. First, note that the bound of (\ref{eqn:reg-bound})  for the regular case can be made arbitrarily close to $2/\pi$. Regarding the bound of (\ref{eqn:nonreg-bound}) for the non-regular case observe that
 $$ \wh{f'}(1) - 2\tau > c_1 \delta - c_1 \delta/2 = c_1 \delta/2$$
which means that the advantage over $1/2$ is at least
$$ (1/2)\cdot c_1 \delta^2 \kappa(\tau) - (1/4) \cdot c_1 \delta \kappa(\tau)^2- 4\sqrt{\eta(\tau)} $$
which is lower bounded by a universal positive constant, since the
second and the third terms are  negligible compared to the first for
our choice of parameters.
This concludes the proof of Theorem~\ref{thm:w1}.



\section{Proof of Theorem~\ref{thm:w1-algo}:  An approximation algorithm for $\w^{\leq 1}[\ltf]$}

\ifnum\short=0
Our approach heavily uses Gaussian analysis, so we record some basic definitions
and facts that we will need below.

\subsection{Gaussian Facts} \label{ssec:gaussian-basics}

\begin{definition} \label{def:normal}
We write $\phi$ for the probability density function of a standard (i.e. zero mean, unit variance) Gaussian; i.e. $\phi(t) = (2\pi)^{-1/2} e^{-t^2/2}$.
We denote by $\Nd(0,1)$ the corresponding distribution and by $\Nd(0,1)^n$ (or $\Nd^n$) the standard $n$-dimensional Gaussian distribution.
\end{definition}

\begin{fact} (Rotational Invariance)  Let $U : \mathbb{R}^n \rightarrow \mathbb{R}^n$ be a unitary transformation, i.e., $U^t U = I$.
If $x \sim \Nd(0,1)^n$, then $Ux \sim \Nd(0,1)^n$.
\end{fact}

\begin{definition}
Let $h_{\theta} : \R \to \bits$ denote the function of one Gaussian random variable $x$ given by $h_{\theta}(x) =\sgn(x-\theta)$.
\end{definition}

\begin{definition}
The function $\mu : \R \cup \{\pm \infty\} \to [-1,1]$ is defined as $\mu (\theta) = \E_{x \sim \Nd(0,1)}[h_{\theta}(x)]$.
Explicitly, $\mu(\theta) = -1 + 2 \int_{\theta}^{\infty} \phi(x) dx $.
We note that $\mu$ is strictly monotone decreasing, hence invertible on $[-1,1]$.
\end{definition}

\begin{definition}
The function $W : [-1,1] \to [0,2/\pi]$ is defined by
$W(x) = (2 \phi(\mu^{-1}(x)))^2$. Equivalently, $W$
is  defined so that $W(\mu(\theta)) =  (2 \phi(\theta))^2$.
\end{definition}

\noindent The next two facts appear as Propositions 24 and 25 in~\cite{MORS:10} respectively.

\begin{fact}\label{fact:gaussian-functions}
Let $X \sim \Nd(0,1)$. We have:  

\begin{itemize}

\item{(i)} $\E [ |X - \theta| ] = 2\phi(\theta) - \theta \mu(\theta)$,
\item {(ii)} $|\mu'| \leq \sqrt{2/\pi}$ everywhere and $|W'| < 1$ everywhere, 
and
\item {(iii)} If $|\nu| = 1-\eta$, then $W(\nu) = \Theta(\eta^2 \log(1/\eta)).$
\end{itemize}
\end{fact}

\begin{fact} \label{fact:gaussian-hermite}
Let $f(x) = \sgn (w \cdot x- \theta)$ be an LTF such that $\|w\|_2 = 1$. Then
\begin{itemize}

\item {(i)} $\tilde{f}(0) \eqdef \E_{x \sim \Nd^n} [f(x)] = \mu(\theta)$,
\item {(ii)} $\tilde{f}(i) \eqdef \E_{x \sim \Nd^n} [f(x) x_i]  = 
\sqrt{W(\tilde{f}(0))} w_i$, for all $i \in [n]$, and
\item {(iii)} $\littlesum_{i=1}^n \tilde{f}^2(i) = W(\tilde{f}(0)).$

\end{itemize}
\end{fact}

\subsection{Proof of Theorem~\ref{thm:w1-algo}}

We recall the statement of Theorem~\ref{thm:w1-algo}:

\medskip

\noindent {\bf Theorem~\ref{thm:w1-algo}.}
\emph{
There is an algorithm that, on input an accuracy parameter $\eps>0$,
runs in time $2^{\poly(1/\eps)}$ and outputs a value $\Gamma_\eps$ such that
$$
\w^{\leq 1}[\ltf]   \leq \Gamma_\eps \leq  \w^{\leq 1}[\ltf]+\eps.
$$
}

\noindent We recall the simple algorithm used to prove Theorem~\ref{thm:w1-algo} from Section~\ref{ssec:w1algo-techniques}:

\begin{quote}
Let $K = \Theta(\eps^{-24}).$ Enumerate all $K$-variable zero-threshold
LTFs, and output the value $$\Gamma_\eps \eqdef \min\{\w^1[f] : f
\text{~is a zero-threshold $K$-variable LTF.} \}.$$
\end{quote}

\fi

As described in Section~\ref{ssec:w1algo-techniques}, it suffices to prove
that for any zero-threshold $n$-variable LTF $f(x)=\sign(w \cdot x),$ there is a $K$-variable
zero-threshold LTF $g$, where $K=\Theta(\eps^{-24})$, such that
\begin{equation} \label{eqn:w1-goal-redux}
|\w^1[f] - \w^1[g]|<\eps;
\end{equation}
we now proceed with the proof.
We can of course assume that $n>K$, since otherwise (\ref{eqn:w1-goal-redux}) is trivially satisfied for $g=f$ with $\eps=0.$

\ignore{

%
%
%
%
%


}

We choose a parameter $\delta = O(\eps^6)$; as described in Section~\ref{ssec:w1algo-techniques},
the proof is by case analysis on the value of the $\delta$-critical index $c(w, \delta)$ of the weight vector $w$.\ignore{
%
} Consider a parameter $L = L(\delta) = \tilde{\Theta}(\delta^{-2}).$ We consider the following two cases:

\medskip

\noindent {\bf [Case I: Large critical index, i.e. $c(w, \delta) \geq L(\delta)$]} In this case, the proof follows easily from the following lemma:

\begin{lemma}[Case II(a) of Theorem 1 of \cite{Servedio:07cc}]
Let $f(x) = \sign(w \cdot x) = \sign(w_H \cdot x_H+w_T \cdot x_T)$, where $w$ is proper, $H = [L(\delta)]$ and $T = [n]\setminus H.$ If $c(w, \delta) \geq L(\delta)$, then $f$ is $\delta$-close
in Hamming distance to function the junta $g(x) = \sign(w_H \cdot x_H).$
\end{lemma}

Since $\dist(f,g) \leq \delta$,
\ifnum\short=0
Fact~\ref{fact:close} implies that
\fi
\ifnum\short=1
the Cauchy-Schwarz inequality can be used to show that
\fi
$|\w^1[f]  - \w^1[g]| < 4\sqrt{\delta} < \eps$. Noting that $g$ is a zero-threshold
 $K$-variable LTF (since $L<K$) completes the proof of Case I.



\medskip

\noindent {\bf [Case II: Small critical index, i.e. $c(w, \delta) < L(\delta)$]} This case requires an elaborate analysis: at a high-level we apply a variable reduction technique
to obtain a junta $g$ that closely approximates the degree-$1$ Fourier weight of $f$. Note that there is no guarantee (and it is typically not the case) that $f$ and $g$ are close
in Hamming distance. Formally, we prove the following theorem:

\begin{theorem} \label{thm:large-ci-w1}
Let $f(x) = \sign(w \cdot x) = \sign(w_H \cdot x_H+w_T \cdot x_T)$, where $w$ is proper, $H = [c(w, \delta)]$ and $T = [n]\setminus H.$ Consider the LTF $g: \bits^{|H|+M} \to \bits$, with $M = \Theta(\eps^{-24})$, defined by
$$g(x_H, z) = \sgn\left( w_H x_H + \|w_T\|_2 \cdot \littlesum_{i=1}^M \frac{z_i}{\sqrt{M}} \right).$$ Then
$|\w^1[f]  - \w^1[g] | < \eps.$
\end{theorem}

Note that $g$ depends on $|H| + M \leq L+M \leq K$ variables. Hence, Theorem~\ref{thm:large-ci-w1} completes the analysis of Case II.
\ifnum\short=0
We refer the reader to Section~\ref{ssec:w1algo-techniques}
for intuition and motivation behind Theorem~\ref{thm:large-ci-w1} and proceed to its proof in the
next subsection.
\fi
\ifnum\short=1
We sketch the proof of Theorem~\ref{thm:large-ci-w1} in the next subsection
\newa{(see the full version for a complete proof)}.
\fi

\ignore{
%
%
%
%
}

\subsection{Proof of Theorem~\ref{thm:large-ci-w1}} \label{ssec:w1-algo-proof}

Let $f:\bn \to \bits$ where $f(x) = \sgn(w \cdot x) = \sgn(w_H x_H + w_T x_T)$, where the tail vector $w_T$ is $\delta$-regular.
Assume wlog that $\|w_T\|_2=1.$ We proceed in the following three steps, which together
yield Theorem~\ref{thm:large-ci-w1}.

\medskip

\noindent {\bf Step 1: ``Gaussianizing'' the tail.} First some notation:  we write
${\cal U}_n$ to denote the uniform distribution over $\bn.$
\ifnum\short=1
We write ${\cal N}^n$ to denote the standard $n$-dimensional Gaussian distribution.
\fi
Our main result in this case is the following theorem, which roughly says that letting tail variables
take Gaussian rather than Boolean values does not change the ``degree-1 Fourier coefficients''
by much:

\begin{theorem} \label{thm:w1-step1}
Let $f = \sign (w_H \cdot x_H + w_T \cdot x_T)$. For $i \in [n]$ define  $ \wh{f} (i) = \E_{x  \sim \mathcal{U}_n} [f(x)x_i]$
and $ \tilde{f} (i) = \E_{x_H  \sim \mathcal{U}_{|H|}, x_T \sim \mathcal{N}^{|T|}} [f(x)x_i]$. If $w_T$ is $\tau$-regular then
\ifnum\short=0
$$\littlesum_{i=1}^n \left( \wh{f}(i) - \tilde{f}(i) \right)^2 = O(\tau^{1/6}).$$
\fi
\ifnum\short=1
$\littlesum_{i=1}^n \left( \wh{f}(i) - \tilde{f}(i) \right)^2 = O(\tau^{1/6}).$
\fi
\end{theorem}

\noindent
\ifnum\short=0
Note that by applying Fact~\ref{fact:close}
\fi
\ifnum\short=1
Using the Cauchy-Schwarz inequality,
\fi
the above theorem implies that $\w^1[f] \approx^{\tau^{1/12}} \widetilde{\w^{1}}[f]$, where
we define $\widetilde{\w^{1}}[f] \eqdef \littlesum_{i=1}^n (\tilde{f} (i))^2.$

\ifnum\short=0

\smallskip

To prove Theorem~\ref{thm:w1-step1} we need a few lemmas.
Our first lemma shows that for a regular LTF, its degree-$1$
Fourier coefficients are close to its corresponding Hermite coefficients.

\begin{lemma} \label{lem:reg-fourier-vs-hermite}
Let $f(x) = \sgn(w \cdot x - w_0)$ be an LTF.
For $i \in [n]$ define $\wh{f}(i) \eqdef \E_{x \in \mathcal{U}_n} [f(x)x_i]$ and  $\tilde{f}(i) \eqdef \E_{x \in \mathcal{N}^n} [f(x)x_i]$.
If $w$ is $\tau$-regular, then $\littlesum_{i=1}^n ( \wh{f}(i) - \tilde{f}(i) )^2 = O(\tau^{1/6})$.
\end{lemma}

\begin{proof}
We can assume that $\|w\|_2=1$. Since $w$ is $\tau$-regular, by Fact~\ref{fact:gaussian-vs-reg} (i) we have that $\wh{f}(0) \approx^{\tau} \tilde{f}(0)$.
It suffices to show that
\begin{equation} \label{eqn:goal}
\littlesum_{i=1}^n \wh{f}(i)^2 + \littlesum_{i=1}^n \tilde{f}(i)^2 \approx^{\tau^{1/6}}  2 \littlesum_{i=1}^n \wh{f}(i) \tilde{f}(i).
\end{equation}

We first note that the lemma follows easily for the case that $|w_0| > \sqrt{2\ln(2/\tau)}.$ In this case, by an application of the Hoeffding bound (Theorem~\ref{thm:chb})
it follows that $|\wh{f}(0)| \geq 1-2\tau$, hence $|\tilde{f}(0)| \geq 1-3\tau$.  By Parseval's identity we have
$$\littlesum_{i=1}^n \wh{f}(i)^2 \leq \littlesum_{\emptyset \neq S \subseteq [n]} \wh{f}(S)^2 = 1 - \wh{f}(0)^2 \leq 4\tau;$$
similarly,  in the Gaussian setting, we get 
$$\littlesum_{i=1}^n \tilde{f}(i)^2 \leq \littlesum_{\emptyset \neq S} \tilde{f}(S)^2 \leq 1 - \tilde{f}(0)^2 \leq 6\tau.$$
Hence, we conclude that $\littlesum_{i=1}^n ( \wh{f}(i) - \tilde{f}(i) )^2 \leq  2 \littlesum_{i=1}^n \wh{f}(i)^2 +  2 \littlesum_{i=1}^n \tilde{f}(i)^2 = O(\tau).$

We now consider the case that $|w_0| \leq \sqrt{2\ln(2/\tau)}$ and proceed to prove (\ref{eqn:goal}).
By Fact~\ref{fact:gaussian-hermite} (iii) we get $\littlesum_{i=1}^n \tilde{f}(i)^2 = W(\tilde{f}(0)).$ Moreover, Theorem~\ref{thm:reg-w1} gives
that $\littlesum_{i=1}^n \wh{f}(i)^2  \approx^{\tau^{1/6}} W(\wh{f}(0)).$ We now claim that $W(\wh{f}(0))  \approx^{\tau} W(\tilde{f}(0)).$ This follows from the mean value theorem,
since $\wh{f}(0) \approx^{\tau} \tilde{f}(0)$ and $|W'| <1$ everywhere, by Fact~\ref{fact:gaussian-functions}. Therefore, we conclude that the LHS of~(\ref{eqn:goal}) satisfies
$$ \littlesum_{i=1}^n \wh{f}(i)^2 + \littlesum_{i=1}^n \tilde{f}(i)^2  \approx^{\tau^{1/6}}  2 W(\tilde{f}(0)).$$
For the RHS of (\ref{eqn:goal})  we can write
$$
\littlesum_{i=1}^n \wh{f}(i) \tilde{f}(i) = \sqrt{W(\tilde{f}(0))}  \littlesum_{i=1}^n w_i \wh{f}(i)
=  \sqrt{W(\tilde{f}(0))}  \E_{x \in \mathcal{U}_n} [ (w \cdot x) \sgn (w \cdot x - w_0)]
$$
where the first equation follows from Fact~\ref{fact:gaussian-hermite} (ii) and the third is Plancherel's identity. Moreover, by definition we have
$$   \E_{x \in \mathcal{U}_n} [ (w \cdot x) \sgn (w \cdot x - w_0)] =    \E_{x \in \mathcal{U}_n} [ |w \cdot x-w_0|] + w_0  \E_{x \in \mathcal{U}_n} [f(x)].$$
Recalling Fact~\ref{fact:gaussian-vs-reg} we deduce that
$$   \E_{x \in \mathcal{U}_n} [ (w \cdot x) \sgn (w \cdot x - w_0)]  \approx^{(|w_0|+1)\tau}    \E_{x \in \mathcal{N}^n} [ (w \cdot x) \sgn (w \cdot x - w_0)] .$$
Now, the RHS above satisfies
\begin{eqnarray*} 
\E_{x \in \mathcal{N}^n} [ (w \cdot x) \sgn (w \cdot x - w_0)]  &=&  
\E_{X \in \mathcal{N}} [ |X-w_0|] +w_0  \E_{x \in \mathcal{N}^n} [f(x)]\\
& = & 2\phi(w_0) = \sqrt{W(\tilde{f}(0))}
\end{eqnarray*}
where the first equality follows by definition, the second uses Fact~\ref{fact:gaussian-functions} (i) and the third uses the definition of $\phi$. Therefore,
$$  \E_{x \in \mathcal{U}_n} [ (w \cdot x) \sgn (w \cdot x - w_0)]  \approx^{(w_0+1)\tau}   \sqrt{W(\tilde{f}(0))}.$$
Since the function $W$ is uniformly bounded from above by $2/\pi$, we conclude that the RHS of~(\ref{eqn:goal}) satisfies
$$ \littlesum_{i=1}^n \wh{f}(i) \tilde{f}(i)  \approx^{(|w_0|+1)\tau}  W(\tilde{f}(0)).$$
The proof now follows from the fact that $(|w_0|+1) \tau < \tau^{1/6}$, 
which holds since $|w_0| = O(\sqrt{\log(1/\tau)})$.
\end{proof}

Our next lemma, a simple generalization of Lemma~\ref{lem:reg-fourier-vs-hermite} above, shows that for any LTF,
if the variables in its tail are replaced by independent standard Gaussians, the corresponding degree $1$-Fourier and Hermite
coefficients of the tail variables are very close to each other.

\begin{lemma}\label{lem:tail-fourier-vs-hermite}
Let $f = \sgn(w_H \cdot x_H + w_T \cdot x_T)$.
For $i \in T$, define $ \wh{f} (i) = \E_{x  \sim \mathcal{U}_n} [f(x)x_i]$ and
$ \tilde{f} (i) = \E_{x_H  \sim \mathcal{U}_{|H|}, x_T \sim \mathcal{N}^{|T|}} [f(x)x_i]$.
If $w_T$ is $\tau$-regular, then we have 
$\littlesum_{i \in T} (\wh{f}(i) - \tilde{f}(i))^2  =  O(\tau^{1/6}).$
\end{lemma}

\begin{proof}
Fix an assignment $\rho \in \{-1 , 1\}^{|H|}$ to the variables in $H$ (head coordinates) and consider the restriction $f_{\rho}$ over the coordinates in $T$,
i.e. $f_{\rho} (x_T) = \sgn(w_H \cdot \rho + w_T  \cdot x_T).$ For every assignment $\rho$, the restriction $f_{\rho}$
is a $\tau$-regular LTF (with a different threshold); hence Lemma~\ref{lem:reg-fourier-vs-hermite} yields that
for all $\rho \in \{-1 , 1\}^{|H|}$ we have
\begin{equation} \label{eqn:regular-restriction}
\littlesum_{i \in T}  \left( \wh{f_{\rho}}(i) -  \tilde{f_{\rho}}(i) \right)^2 = O(\tau^{1/6}).
\end{equation}
Hence, we obtain
\begin{eqnarray*}
\littlesum_{i \in T } \left( \wh{f}(i) - \tilde{f}(i) \right)^2  =  \littlesum_{i \in T } {\left ( \E_{\rho \sim \mathcal{U}_{|H|}}  \left[ \wh{f_{\rho}}(i) -  \tilde{f_{\rho}}(i) \right]    \right)^2 }
&\leq&  \littlesum_{i \in T} { \E_{\rho \sim \mathcal{U}_{|H|}}  \left[ \left( \wh{f_{\rho}}(i) -  \tilde{f_{\rho}}(i) \right)^2 \right]   } \\
&=& \E_{\rho \sim \mathcal{U}_{|H|}}   \left[    \littlesum_{i \in T}  \left( \wh{f_{\rho}}(i) -  \tilde{f_{\rho}}(i)  \right)^2  \right]  \\
&=&  O(\tau^{1/6})
\end{eqnarray*}
where the first equality uses the definition of the Fourier/Hermite coefficients, the first inequality follows from Jensen's inequality for each summand,
the second equality follows by linearity and the last equality uses (\ref{eqn:regular-restriction}).
\end{proof}

Replacing the Boolean tail variables by Gaussians alters the Fourier 
coefficients of the head variables as well.
Our next lemma shows that the corresponding change is bounded in terms of 
the regularity of the tail.

\begin{lemma} \label{lem:head-fourier-vs-hermite}
Let $f = \sgn(w_H \cdot x_H + w_T \cdot x_T)$. For $i \in H$ define $ \wh{f} (i) = \E_{x  \sim \mathcal{U}_n} [f(x)x_i]$
and $ \tilde{f} (i) = \E_{x_H  \sim \mathcal{U}_{|H|}, x_T \sim \mathcal{N}^{|T|}} [f(x)x_i]$.
If $w_T$ is $\tau$-regular, then we have $\littlesum_{i \in H} (\wh{f}(i) - \tilde{f}(i))^2  = O( \tau^2).$
\end{lemma}

\begin{proof}
We define the functions $f' : \{-1,1\}^{|H|} \to [-1,1]$ and
$f'' : \{-1,1\}^{|H|} \to [-1,1]$ as follows :
$$
f'(x_H) = \E_{x_T \in \mathcal{U}_{|T|}} \left[ f(x_H, x_T) \right]  \textrm{ and }
f''(x_H) = \E_{x_T \in\mathcal{N}^{|T|}} \left[ f(x_H, x_T) \right].
$$
By definition, for all $i \in H$ it holds $\wh{f'}(i) = \wh{f}(i)$ and $\wh{f''}(i) = \tilde{f}(i)$.
We can therefore write
\begin{eqnarray*}
\littlesum_{i \in H} (\wh{f}(i) - \tilde{f}(i))^2 &=& 
\littlesum_{i \in H} (\wh{f'}(i) - \wh{f''}(i))^2\\
&\leq& \littlesum_{S \subseteq H}  (\wh{f'}(S)  - \wh{f''}(S))^2 = \E_{x \in \mathcal{U}_{|H|} }  (f'(x) - f''(x))^2 \le \| f' - f'' \|_{\infty}^2
\end{eqnarray*}
where the second equality is Parseval's identity and the final inequality follows from the monotonicity of the norms
($\|\cdot\|_{\infty}$ denotes the sup-norm of a random variable).

In order to bound $ \| f' - f'' \|_{\infty}$ we exploit the regularity of the tail via the Berry-Ess{\'e}en theorem.
Indeed, fix an assignment $\rho \in \bits^{|H|}$ to $x_H$. Then
\begin{eqnarray*}
\left| f' (\rho) - f''(\rho) \right| &\le& 
2 \left| \Pr_{x_T \in  \{-1,1\}^{|T|}}  \left[ w_T \cdot x_T  + w_H \cdot  \rho  \geq 0 \right] - \right.\\
& & \left.\text{~~~~~~~~~~}\Pr_{x_T \in  \mathcal{N}^{|T|}} \left [w_T \cdot x_T + w_H \cdot \rho  \geq 0 \right] \right|
\end{eqnarray*}
Since $w_T$ is $\tau$-regular, by Fact~\ref{fact:be}, the RHS above is bounded from above by $2\tau$. Since this holds for any restriction $\rho$ to the head we conclude that
$ \| f' - f'' \|_{\infty} \leq 2\tau$ as desired.
\end{proof}

\noindent Theorem~\ref{thm:w1-step1} follows by combining Lemmas~\ref{lem:tail-fourier-vs-hermite} and~\ref{lem:head-fourier-vs-hermite}.

\fi

\medskip

\noindent {\bf Step 2: ``Collapsing'' the tail.}
Let $F : \bits^{|H|} \times  \R \to \bits$ be defined by
$F(x_H, y) = \sgn(w_H x_H + y)$
\newa{(recall that we have assumed that $w$ is scaled so that
the ``tail weight'' $\|w_T\|_2$ equals 1).}
 For $i \in H$, we define
$$\wh{F}(i) = \E_{x_H \sim \mathcal{U}_{|H|}, y \sim \mathcal{N}(0,1)} [F(x_H, y) x_i]$$ and
$$\wh{F}(y) = \E_{x_H \sim \mathcal{U}_{|H|}, y \sim \mathcal{N}(0,1)} [F(x_H, y) y].$$
We also denote $\widetilde{\w^1}[F] = \littlesum_{i \in H} (\wh{F}(i))^2 + \wh{F}(y)^2.$
Our main result for this step is that ``collapsing'' all $|T|$ tail Gaussian variables
to a single Gaussian variable does not change the degree-1 ``Fourier weight'':

\begin{theorem} \label{thm:w1-step2}
We have that $\widetilde{\w^1}[F] =  \widetilde{\w^1}[f].$
\end{theorem}

\ifnum\short=0

The theorem follows by combining the following two lemmas.

\begin{lemma}  \label{lem:head-step2}
For every $i \in H$, $\tilde{f}(i) =\wh{F}(i)$.
\end{lemma}

\begin{proof}
The lemma follows straightforwardly by the definitions. Indeed,
for every $i \in H$,
\begin{eqnarray*}
\tilde{f}(i) &=& \E_{x_H \sim \mathcal{U}_{|H|}, x_T \sim 
\mathcal{N}^{|T|}} [ \sgn (w_H x_H + w_T  x_T) x_i ]\\
& = &
\E_{x_H \sim \mathcal{U}_{|H|}, y \sim \mathcal{N}(0,1)} 
[ \sgn (w_H x_H + y ) x_i ]  = \widehat{F}(i)
\end{eqnarray*}
where the third equality uses the fact 
that $w_T \cdot x_T$ is distributed as $\mathcal{N}(0,1)$.
\end{proof}

\begin{lemma}\label{lem:tail-step2}
We have that $(\wh{F}(y))^2  = \littlesum_{i \in T} (\tilde{f}(i))^2.$
\end{lemma}
\begin{proof}
This lemma is intuitively clear but we nonetheless give a proof.
We need the following simple propositions.

\begin{proposition}\label{prop:unitary}
Let $h : \R^m \to \R$ with $h \in L_2(\mathcal{N}(0,1)^m)$.
Let  $U : \R^m \to \R^m$ be a unitary linear transformation.
For $i \in [m]$, define $\tilde{h}(i) =  \E_{x \sim \mathcal{N}^m} [h(x)x_i]$
and $\tilde{h}(i)' =  \E_{x \sim \mathcal{N}^m} [ h(x) (Ux)_i ].$
Then, $\littlesum_{i=1}^m \tilde{h}(i)^2 = \sum  _{i=1}^m \tilde{h}(i)'^2$.
\end{proposition}

\begin{proof}
Let $(Ux)_i = \littlesum_{j=1}^m a_{ij} x_j$. By linearity, we get that $\tilde{h}(i)'  =  \littlesum_{j=1}^m a_{ij} \tilde{h}(j)$. Then,
$$
\littlesum_{i=1}^m \tilde{h}(i)'^2 = \littlesum_{j=1}^m \left( \littlesum_{i=1}^m a_{ij}^2 \right) \tilde{h}(j)^2  +   \littlesum_{j \neq i} \left( \littlesum_{k=1}^m a_{kj} a_{ki} \right) \tilde{h}(i) \tilde{h}(j)
$$
By elementary properties of unitary matrices, we have 
(i) $\littlesum_{i=1}^m a_{ij}^2=1$ for all $j$, and (ii) 
$\littlesum_{k=1}^m a_{kj}a_{ki}=0$
for $i \neq j$.
Substitution completes the proof.
\end{proof}

\begin{proposition}\label{prop:gaussian-step2}
Let $\Psi: \R \to \R$ and $\Phi : \R^{m} \to \R$ with $\Phi  \in L_2(\mathcal{N}(0,1)^m)$ defined as
$\Phi(x) = \Psi \left(\littlesum_{i=1}^{m} w_i x_i\right)$,
where $x, w \in \R^m$ with $\| w \|_2 =1$.
Then
$\littlesum_{i=1}^{m} \widetilde{\Phi} (x_i)^2 =\widetilde{\Psi}(y)^2.$
\end{proposition}

\begin{proof}
It is clear there is a unitary matrix $U$ such that
$ U : x_1 \mapsto \littlesum_{i=1}^m w_i x_i$. Hence,
an application of Proposition~\ref{prop:unitary} gives us
\begin{equation}\label{eq:1}
\littlesum_{i=1}^{m} \tilde{\Phi}(x_i)^2 = \littlesum_{i=1}^{m} \left( \E_{x \in \mathcal{N}^m}  [ \Phi(x) \cdot (U x)_i ]  \right)^2
\end{equation}
Now observe that for $i>1$, $\E_{x \in \mathcal{N}^{m}}  [ \Phi(x) \cdot (U x)_i ]  = 0$
as $\Phi(x)$ is independent of $(U x)_i$.
Using the rotational invariance of the Gaussian measure and using $y$ instead of $(Ux)_1$ we deduce
$$
 \E_{x \in \mathcal{N}^{m}}   [\Phi(x) \cdot (U x)_1] = \E_{y \in \mathcal{N}(0,1)}  [ (\Psi(y)  y ]  =\widetilde{\Psi}(y).
$$
Combining with (\ref{eq:1}) completes the proof.
\end{proof}

The proof of the lemma follows by a simple application of the above proposition.
Indeed, set
$\Psi(y) = \E_{x_H \sim \mathcal{U}_{|H|}} [ F(x_H, y)].$
 An application of Proposition~\ref{prop:gaussian-step2} gives us
 $\littlesum_{i \in T} \widetilde{\Phi}(x_i) ^2 = \widetilde{\Psi}(y)^2.$
 Now note that for $i \in T$, by definition,
 $\widetilde{\Phi}(x_i) = \widetilde{f}(i)$,
 and $\widetilde{\Psi}(y) =  \wh{F}(y)$.
 This completes the proof of the lemma.
\end{proof}

\fi

\medskip

\noindent {\bf Step 3: ``Booleanizing'' the tail.}
 Let $M  =  \Theta (\epsilon^{-24})$. Consider the LTF
$g$ mapping $(x_H, z) \to \bits$, where $x_H \in \bits^{|H|}$ 
and $z \in \bits^{M}$, defined by
$$ g(x_H, z) = \sgn \left(w_H \cdot x_H + (\littlesum_{i=1}^{M} z_i) / \sqrt{M}\right).$$

In this step we show that replacing the (single) Gaussian tail variable with a scaled sum of
Boolean variables does not change the degree-1 ``Fourier weight'' by much:

\begin{theorem} \label{thm:w1-step3}
We have that $| \widetilde{\w^{1}}[F] - \w^1[g] | = O(M^{-1/24}).$
\end{theorem}

\ifnum\short=0

As expected the theorem follows by combining two lemmas, one to deal
with the head and one for the tail.

\begin{lemma}\label{lem:head-step3}
We have $\littlesum_{i \in H} (\wh{F}(i) -\wh{g}(i))^2 = O(M^{-1}).$
\end{lemma}

\begin{proof}
The proof closely parallels that of Lemma~\ref{lem:head-fourier-vs-hermite}.
Namely, we will define the function $h, h' : \bits^{|H|} \rightarrow [-1,1]$ as
$h(x_H) = \E_{y \sim \mathcal{N}(0,1)} [F(x_H,y)] $ and $h' (x_H) = \E_{z \sim \mathcal{U}_{M}} [g(x_H,z)].$
Note that for $i \in H$, $\wh{h}(i) = \wh{F}(i)$ and $\wh{h'}(i)  = \wh{g}(i)$.
As in Lemma~\ref{lem:head-fourier-vs-hermite}, we have
$\littlesum_{i \in H} (\wh{h}(i)  - \wh{h'}(i)  )^2 \leq \| h- h' \|_{\infty}^2$.
For any $\rho \in \bits^{|H|}$, we can write
$$|h (\rho) - h'(\rho)| = 2 |\Pr_{y \in \mathcal{N}(0,1)} [ w_H \rho + y  \geq 0 ]  - \Pr_{z \in \mathcal{U}_M} [ w_H \rho + \littlesum_{i=1}^M z_i /\sqrt{M} \geq 0] |.$$
Theorem~\ref{thm:be} shows that the RHS  is bounded from above by $2/\sqrt{M}$, which completes the proof of the lemma.
\end{proof}

\begin{lemma}\label{lem:tail-step3}
Let $F$ and $g$ as defined above.
Then
$|\littlesum_{i=1}^M (\wh{g}(z_i))^2 - (\wh{F}(y))^2| = O(M^{-1/24}).$
\end{lemma}

\begin{proof}
First note that, by symmetry, for all $i, j \in [M]$ we have $\wh{g}(z_i)  = \wh{g}(z_j).$
By definition, we can write
$$\wh{g}(z_i) = \E_{x_H} \E_{z \sim \mathcal{U}_M} \left[ \sgn \left(w_H \cdot x_H + (\littlesum_{i=1}^M z_i) / \sqrt{M}\right) \right]$$
and let us also denote
$$\tilde{g}(z_i) = \E_{x_H} \E_{z \sim \mathcal{N}^M} \left[ \sgn \left(w_H \cdot x_H + (\littlesum_{i=1}^M z_i) / \sqrt{M}\right) \right].$$
Since  the tail of $g$ is $1/\sqrt{M}$-regular Lemma~\ref{lem:tail-fourier-vs-hermite} implies that
$$\littlesum_{i=1}^M (\wh{g}(z_i) - \tilde{g}(z_i))^2 = O(M^{-1/12})$$ and by Fact~\ref{fact:close}
$$|\littlesum_{i=1}^M (\wh{g}(z_i))^2   - \littlesum_{i=1}^M (\tilde{g}(z_i))^2|= O(M^{-1/24}).$$
Since $\wh{F}(y) = \E_{x_H} \E_{y \sim \mathcal{N}} \left[ \sgn \left(w_H \cdot x_H + y\right) \right]$,
arguments identical to those of Lemma \ref{lem:tail-step2} give us
$$\littlesum_{i=1}^M (\tilde{g}(z_i))^2 = (\wh{F}(y))^2.$$
This completes the proof.
\end{proof}

\fi




\section{Proof of Theorem~\ref{thm:algo-T}:  An approximation algorithm for $\T(\mathbb{S})$}
\label{sec:alg-T}

In this section we prove Theorem~\ref{thm:algo-T} (restated below):

\medskip

\noindent {\bf Theorem~\ref{thm:algo-T}.}
\emph{
There is an algorithm that, on input an accuracy parameter $\eps>0$,
runs in time \newa{$2^{{\poly(1/\eps)}}$} and outputs a value $\Gamma_\eps$ such that
\[
\T(\mathbb{S})   \leq \Gamma_\eps \leq  \T(\mathbb{S})+\eps.
\]
}

The main structural result required to prove Theorem~\ref{thm:algo-T} is the following theorem
(recall that $\mathbb{S}^{n-1}$ denotes the unit sphere in $\R^n$,
i.e. $\mathbb{S}^{n-1} = \{x \in \R^n: \|x\|_2=1\}$):
\begin{theorem}\label{thm:T-dog}
For any $\epsilon>0$, there is a value $K_{\epsilon} = \poly(1/\epsilon)$ such that for any $n \in \mathbb{N}$,
$$
\T(\mathbb{S}^{n-1}) \le \T(\mathbb{S}^{K_{\epsilon}-1})  \le \T(\mathbb{S}^{n-1}) + \epsilon.
$$
As a corollary, we have $\T(\mathbb{S}) \le \T(\mathbb{S}^{K_{\epsilon}-1})  \le \T(\mathbb{S}) + \epsilon$.
\end{theorem}
Theorem~\ref{thm:T-dog} implies that to compute $\T(\mathbb{S}) $ up to accuracy $\epsilon$, it suffices to compute $\T(\mathbb{S}^{K_{\epsilon}-1})$; i.e.,  we need to compute $\inf_{w \in \mathbb{S}^{K_{\epsilon} -1}} \T(w)$. While $\mathbb{S}^{K_{\epsilon} -1}$ is a finite-dimensional object, it is an (uncountably) infinite set and hence it is not immediately obvious  how to compute $\inf_{w \in \mathbb{S}^{K_{\epsilon} -1}} \T(w)$. The next lemma says that this can indeed be computed in time \newa{$2^{\widetilde{O}(K_{\epsilon}^2 )}$}.

\begin{lemma}\label{lem:compute-T}
For any $m \in \mathbb{N}$, $\T(\mathbb{S}^{m-1})$ can be computed exactly in time \newa{$2^{\tilde{O}(m^2)}$}.
\end{lemma}
Theorem~\ref{thm:algo-T} follows by combining Theorem~\ref{thm:T-dog} and Lemma~\ref{lem:compute-T}.

\subsection{Proof of Theorem~\ref{thm:T-dog}}
\begin{proof}[Proof of Theorem~\ref{thm:T-dog}]
Let $w \in \mathbb{S}^{n-1}$. For $\epsilon>0$, we will prove that there exists a value $K_{\epsilon} = O(1/\epsilon^3)$ and $v \in \mathbb{S}^{K_{\epsilon}-1}$ such that $|\T(v)  - \T(w)|\le \epsilon$. Clearly, the upper bound on $ \T(\mathbb{S}^{K_{\epsilon}-1})$ in Theorem~\ref{thm:T-dog} follows from this. The lower bound on $ \T(\mathbb{S}^{K_{\epsilon}-1})$ is obvious.

To prove the existence of vector $v \in  \mathbb{S}^{K_{\epsilon}-1}$, we begin by considering the $\eta$-critical index of $w$ for $\eta= \epsilon/64$.
We also let $K = C \cdot t/\eta^2   \cdot \log (t/\eta)$  where $t$ will be chosen later to be $O(\log (1/\eta))$ and $C$ to be a sufficiently large constant. Clearly, for this choice of $\eta$ and $t$, we have that $K =O(1/\epsilon^3)$. The next two claims show that whether $\newa{c(w,\eta)}$,
the $\eta$-critical index of $w$, is larger or smaller than $K$, the desired vector $v$
exists in either case.
\begin{claim}\label{clm:large-critical-index}
Let $w \in \mathbb{S}^{n-1}$ be such that $c(w,\eta)>K$. Then there is a vector $v \in \mathbb{S}^{K}$
such that $|\T(v) - \T(w)| \le \eta$.
\end{claim}
\begin{claim}\label{clm:small-critical-index}
Let $w \in \mathbb{S}^{n-1}$ be such that $c(w,\eta) \leq K$. Then there exists $v \in \mathbb{S}^{K  +\lambda(\eta)-1}$ such that
$
|\T(v) - \T(w)| \le \newa{8}\eta
$,
where $\lambda(\eta) = 4/\eta^2$.
\end{claim}
In both Claim~\ref{clm:large-critical-index} and Claim~\ref{clm:small-critical-index}, the final vector $v$ is at most $K+ \lambda(\eta) \le O(1/\epsilon^3)$-dimensional. Hence Theorem~\ref{thm:T-dog} follows by choosing $K_{\epsilon} = O(1/\epsilon^3)$.
\end{proof}
We start by proving Claim~\ref{clm:large-critical-index}. We will require the following anti-concentration
lemma from \cite{OS11:chow}:
\begin{lemma}
\label{lem:chow-OS}(Theorem~4.2 in \cite{OS11:chow})
Let $w \in \mathbb{S}^{n-1}$, $0 <\eta <1/2$, $t>1$  and let $K$ be defined (in terms of $t$ and $\eta$) as above. If $c(w,\eta)>K$, then for any $w_0 \in \mathbb{R}$, we have
$$
\Pr_{x \in \{-1,1\}^n} \left[\left|\littlesum_{i=1}^n w_i x_i -w_0\right| \le \sqrt{t}  \cdot \sigma_{K} \right] \le  2^{-t}
$$
where $\sigma_{K} = \newa{\|w^{(K)}\|_2} = \sqrt{\sum_{j>K} w_j^2}$.
\end{lemma}
\begin{remark}
Lemma~\ref{lem:chow-OS} as stated in~\cite{OS11:chow} has the probability bounded by $O(2^{-t})$. However, by making $C$ large enough, it is obvious that the probability can be made $2^{-t}$.
\end{remark}
\begin{proof}[{{ Proof of Claim~\ref{clm:large-critical-index}}}]
Choose a specific $w_0$ (we will fix it later).
By Lemma~\ref{lem:chow-OS}, we have that,
\begin{equation}\label{eq:ineq1}
\Pr_{x \in \{-1,1\}^n} \left[\left|\littlesum_{i=1}^n w_i x_i -w_0\right| \le \sqrt{t}  \cdot \sigma_{K} \right] \le  2^{-t}.
\end{equation}
Note that $\littlesum_{i=1}^n w_i x_i = \littlesum_{i \le K} w_i x_i + \littlesum_{i>K} w_i x_i$.
Let $w_T$ denote the ``tail weight vector'' $w_T = (w_{K+1},\dots,w_n).$
Since $\Vert w_T \Vert_2 = \sigma_{K}$, by Hoeffding's inequality
(Theorem~\ref{thm:chb}), we have that
\begin{equation}\label{ineq:2}
\Pr_{x \in \{-1,1\}^{n-K}} \left[\left|w_T \cdot x_T\right| \newa{>} \frac{1}{2} \cdot \sqrt{\frac{t}{2}} \cdot \sigma_{K}  \right]  \le 2^{\frac{-t}{8}}.
\end{equation}
Define the set $A_{\good,w_0}$ as follows :
$$
A_{\good,w_0} := \left\{x \in \{-1,1\}^n : \left|\littlesum_{i=1}^n w_i x_i -w_0\right| \ge \sqrt{t}  \cdot \sigma_{K} \textrm{  and } \left|\littlesum_{i>K} w_i x_i\right| \le \sqrt{\frac{t}{8}} \cdot \sigma_{K} \right\}
$$
We next make a couple of observations about the set $A_{\good,w_0}$. The first is that combining (\ref{eq:ineq1}) and (\ref{ineq:2}), we get that $\Pr_{x \in \{-1,1\}^n} [x \not \in A_{\good,w_0}] \le  2^{\frac{-t}{8}} +  2^{-t} $.
Second, for every $x \in A_{\good,w_0}$, we have
$$\left|\littlesum_{i=1}^{K} w_i x_i - w_0\right| \ge \left|\littlesum_{i=1}^{n} w_i x_i - w_0\right| - \left|\littlesum_{i>K} w_i x_i\right| \ge \sqrt{t} \cdot \sigma_K \cdot \left( 1- \frac{1}{2\sqrt{2}} \right) \ge \sqrt{t} \cdot \sigma_K \cdot  \frac{3}{5}.$$
Hence for every $x \in A_{\good, w_0}$, we have
$$\mathbf{1}_{\littlesum_{i=1}^n w_i x_i \le w_0} = \mathbf{1}_{\littlesum_{i=1}^K w_i x_i \le w_0}.$$

Now, consider the vector $v' \in \mathbb{S}^{n}$ defined as follows:

\begin{itemize}

\item $v'_i=w_i$ for $1 \le i \le K$;

\item $v'_{K+1} = \sigma_K$; and

\item $v'_j = 0$ for $j > K+1.$

\end{itemize}
 Note that for every $x \in \{-1,1\}^n$ (and hence for every $x \in A_{\good,w_0}$),
 we have $\littlesum_{i=1}^K w_i x_i = \littlesum_{i=1}^K v'_i x_i$.  Recalling that \
$$
\left|\littlesum_{i=1}^K v'_i x_i  - w_0\right| = \left| \littlesum_{i=1}^K w_i x_i - w_0 \right| \ge   \sqrt{t} \cdot \sigma_K \cdot  \frac{3}{5}
\quad \text{for $x \in A_{\good, w_0}$}
$$
and that all $x \in \bn$ satisfy  $|\littlesum_{i>K}  v'_i x_i | \le \sigma_K$, for $t$ sufficiently large
we get that
every $x \in A_{\good,w_0}$ satisfies
$$
\mathbf{1}_{\littlesum_{i=1}^n v'_i x_i \le w_0} = \mathbf{1}_{\littlesum_{i=1}^K v'_i x_i \le w_0}.
$$
Thus, for $x \in A_{\good, w_0}$, we have that all four events coincide:
$$
 \mathbf{1}_{\littlesum_{i=1}^n v'_i x_i \le w_0}= \mathbf{1}_{\littlesum_{i=1}^K v'_i x_i \le w_0} =  \mathbf{1}_{\littlesum_{i=1}^K w_i x_i \le w_0}=\mathbf{1}_{\littlesum_{i=1}^n w_i x_i \le w_0}.
$$
Likewise, we also get that for $x \in A_{\good,w_0}$,
$$
 \mathbf{1}_{\littlesum_{i=1}^n v'_i x_i \ge w_0}= \mathbf{1}_{\littlesum_{i=1}^K v'_i x_i \ge w_0} =  \mathbf{1}_{\littlesum_{i=1}^K w_i x_i \ge w_0}=\mathbf{1}_{\littlesum_{i=1}^n w_i x_i \ge w_0}.
$$
Now, let $S  = A_{\good,1} \cap A_{\good,-1}$. We then get that for $x \in S$,
$$
 \mathbf{1}_{\littlesum_{i=1}^n v'_i x_i \in [-1,1] } =   \mathbf{1}_{\littlesum_{i=1}^n w_i x_i \in [-1,1] }.
$$
Since $\Pr_{x \in \{-1,1\}^n} [x \not \in A_{\good,w_0}] \le  2^{\frac{-t}{8}} +  2^{-t} $ for $w_0 \in \{-1,1\}$, as a result we have $\Pr[x \not \in S] \le 2 \cdot (2^{-t/8} + 2^{-t} )$.  Taking $t = 8 \log (16/\eta)$, we get that $\Pr[x \not \in S] \le \eta/4$. This implies that $$\left| \Pr_{x \in \{-1,1\}^n} \left[\left|\littlesum_{i=1}^n w_i x_i\right| \le 1 \right] - \Pr_{x \in \{-1,1\}^n} \left[\left|\littlesum_{i=1}^n v'_i x_i\right| \le 1 \right] \right| \le \eta/4.$$ Since the final
$n-K-1$ coordinates of $v'$ are zero, if we simply
truncate $v'$ to the first $K+1$ coordinates, we get a
vector $v \in \mathbb{S}^{K}$ such that
$$
\left| \Pr_{x \in \{-1,1\}^n} \left[\left|\littlesum_{i=1}^n w_i x_i\right| \le 1 \right] - \Pr_{x \in \{-1,1\}^{K+1}} \left[\left|
\littlesum_{i=1}^{\newa{K+1}} v_i x_i\right| \le 1 \right] \right| \le \eta/4,
$$
and Claim~\ref{clm:large-critical-index} is proved.
\end{proof}

We next move to the proof of  Claim~\ref{clm:small-critical-index}. For that, we will need the following key proposition.
 \begin{proposition}\label{prop:regular-tail}
 Let $w, u \in \mathbb{S}^{n-1}$ be such that $w_i  = u_i$ for  $1 \le i \le K$. Suppose moreover that
 $w^{(K)}\eqdef (w_{K+1},\dots,w_n)$ and $u^{(K)} \eqdef (u_{K+1},\dots,u_n)$ are both $\eta$-regular.  Then, for any $w_0 \in \mathbb{R}$, we have that $$\left|\Pr_{x \in \{-1,1\}^n} [{w \cdot x \le w_0}] - \Pr_{x \in \{-1,1\}^n} [{u \cdot x \le w_0}] \right| \le 4\eta$$
 and
 $$\left|\Pr_{x \in \{-1,1\}^n} [{w \cdot x \ge w_0}] - \Pr_{x \in \{-1,1\}^n} [{u \cdot x \ge w_0}] \right| \le 4\eta.$$
 \end{proposition}
 \begin{proof}
Consider any fixed setting of variables $x_1, \ldots, x_K \in \{-1,1\}$. Note that $\littlesum_{i=1}^n w_{i}  x_i = \sum_{i \le K} w_{i} x_i + \sum_{i>K} w_{i} x_i  $.  We have  \begin{eqnarray*}
 \Pr_{{x^{(K)} \in \{-1,1\}^{n-K}}} \left[\littlesum_{i=1}^n w_{i}  x_i  \le w_0\right] &=& \Pr_{x^{(K)} \in \{-1,1\}^{n-K}} \left[\littlesum_{i>K} w_{i}  x_i  \le w_0 -\littlesum_{i \le K} w_{i} x_i  \right].  \end{eqnarray*}
 However, as $w^{(K)}$ is $\eta$-regular, by Theorem~\ref{thm:be},  we get
\begin{eqnarray*}
&&\left|\Pr_{x^{(K)} \in \{-1,1\}^{n-K}} \left[\littlesum_{i>K} w_{i}  x_i  \le w_0 -\littlesum_{i \le K} w_{i} x_i  \right] \right.\\
&& \left.- \Pr \left[ \mathcal{N}(0, \Vert w^{(K)} \Vert) \le w_0 -\littlesum_{i \le K} w_{i} x_i  \right]  \right| \le 2 \eta.
\end{eqnarray*}
Likewise, we have 
\begin{eqnarray*}
&&\left|\Pr_{x^{(K)} \in \{-1,1\}^{n-K}} \left[\littlesum_{i>K} u_{i}  x_i  \le w_0 -\littlesum_{i \le K} u_{i} x_i  \right]\right.\\
&& \left. - \Pr \left[ \mathcal{N}(0, \Vert u^{(K)} \Vert) \le w_0 -\littlesum_{i \le K} u_{i} x_i  \right]  \right| \le 2 \eta. 
\end{eqnarray*}
As $w_{i} =u_{i}$ for $1\le i \le K$, we get that 
\begin{eqnarray*}
&&
\left|\Pr_{x^{(K)} \in \{-1,1\}^{n-K}} \left[\littlesum_{i>K} u_{i}  x_i  \le w_0 -\littlesum_{i \le K} u_{i} x_i  \right] \right.\\
&&\left.  - \Pr_{x^{(K)} \in \{-1,1\}^{n-K}} \left[\littlesum_{i>K} w_{i}  x_i  \le w_0 -\littlesum_{i \le K} w_{i} x_i  \right]   \right| \le 4 \eta. 
\end{eqnarray*}
As the above equation is true for any setting of $x_1, \ldots, x_K$, we get that
$$\left|\Pr_{x \in \{-1,1\}^n} [w \cdot x \le w_0] - \Pr_{x \in \{-1,1\}^n} [u \cdot x \le w_0] \right| \le 4\eta.$$ The second part of the proposition follows in exactly the same way.
\end{proof}
\begin{proof}[{{ Proof of Claim~\ref{clm:small-critical-index}}}]
Let $w= (w_1, \ldots, w_K, \ldots, w_n)$ where $K' \leq K$ is the $\eta$-critical index of $w$. Construct a new vector $v'$ such that $v'_i = w_i$ for $1 \le i \le K'$. For $1 \le j \le \lambda(\eta) =4/\eta^2$, we let $v'_{i+j} = (\eta/2) \cdot  \Vert w^{(K')} \Vert$, where as before
$w^{(K')}$ denotes the $(n-K')$-dimensional vector $(w_{K'+1},\dots,w_n).$  For $j > \lambda(\eta)$, we define $v'_{i+j}=0$.

 It is clear that $v' \in \mathbb{S}^n$ and that $v'^{(K')}$ is $\eta$-regular. By Proposition~\ref{prop:regular-tail}, we have
$$
\left|\Pr_{x \in \{-1,1\}^n} [{w \cdot x \le 1}] - \Pr_{x \in \{-1,1\}^n} [{v' \cdot x \le 1}] \right| \le 4\eta
$$
and
$$
\left|\Pr_{x \in \{-1,1\}^n} [{w \cdot x \ge -1}] - \Pr_{x \in \{-1,1\}^n} [{v' \cdot x \ge -1}] \right| \le 4\eta.
$$
Combining these two, we get
$$
\left|\Pr_{x \in \{-1,1\}^n} [|w \cdot x| \le 1] - \Pr_{x \in \{-1,1\}^n} [|v' \cdot x| \le 1] \right| \le 8\eta.
$$
As all the coordinates of $v'$ beyond the first $K'+ \lambda(\eta)$ coordinates are zero, if we truncate $v'$ to its first $K'+ \lambda(\eta)$ coordinates, we get $v \in \mathbb{S}^{K'+\lambda(\eta) -1}$ such that
$$
\left|\Pr_{x \in \{-1,1\}^n} [| w \cdot x |\le 1] - \Pr_{x \in \{-1,1\}^{K'+\lambda(\eta)}} [|v\cdot x |\le 1] \right| \le 8\eta
$$
and Claim~\ref{clm:small-critical-index} is proved.  \end{proof}

 \subsection{Proof of Lemma~\ref{lem:compute-T}}

\newa{
The proof of Lemma~\ref{lem:compute-T} that we give below is based on the decidability of the existential theory of the reals.  We believe that it may be possible to prove this lemma without invoking the
existential theory of the reals, by combining
perturbation-based arguments with convex programming.   However, carefully formalizing such arguments is
a potentially involved process, so we have given (what seemed to us to be) a more concise proof,
using the existential theory of reals, below.}

 \begin{proof}[Proof of Lemma~\ref{lem:compute-T}]
We use the following result due to Renegar~\cite{Ren:88}.
 \begin{theorem}\label{thm:first-order}\cite{Ren:88}
 There is an algorithm $\mathcal{A}_{Ren}$ which, given a set of real polynomials $p_1, \ldots, p_m : \mathbb{R}^n \rightarrow \mathbb{R}$ and $q_1, \ldots, q_k : \mathbb{R}^n \rightarrow \mathbb{R}$ with
 rational coefficients,  decides whether there exists an $x \in \mathbb{R}^n $ such that
 \begin{itemize}
 \item $\forall i \in [m]$, $p_ i(x) \ge 0$, and
 \item $\forall i \in [k]$, $q_i (x) >0$.
 \end{itemize}
 If the bit length of all coefficients in all polynomials is at most $L$ and the maximum degree of
 any polynomial is at most $d$, then the running time of $\mathcal{A}_{Ren}$ is $L^{O(1)} \cdot ((m+k)\cdot d)^{O(n)}$.
 \end{theorem}
 The following is an obvious corollary of the above theorem :
 \begin{corollary}\label{corr:first-order}
There is an algorithm  which, given a set $S \subset \{-1,1\}^m$, decides whether there exists a vector $w \in \mathbb{S}^{m-1}$ such that every $x \in S$ has $|w \cdot x| >1$.  The algorithm runs in time $2^{O(m^2)}$.
 \end{corollary}
 \begin{proof}
 Let $p_{0} : \mathbb{R}^m \rightarrow \mathbb{R}$ be defined as $p_0(w) = \sum_{i=1}^m w_i^2-1$. For each
 $x \in S$, define $q_x : \mathbb{R}^m \rightarrow \mathbb{R}$ to be $q_x(w) = (\littlesum_{i=1}^m w_i \cdot x_i)^2-1$.  Consider the following set of constraints (call it $\mathcal{L}$) :
 \begin{itemize}
 \item $p_0(w) \ge 0$.
 \item $-p_0(w) \ge 0$.
 \item $\forall x \in S$, $q_x(w) >0$.
 \end{itemize}
 Clearly, this set of constraints has a solution if and only if there exists some $w \in \mathbb{S}^{m-1}$ such that  $|w \cdot x| >1$ for all $x \in S$.  Note that each of the polynomials in $\mathcal{L}$ is of degree $2$ and all coefficients have constant-size representations.  The total number of constraints
 in ${\cal L}$ is $|S| +2 \le 2^m + 2$. This means that $\mathcal{A}_{Ren}$ can decide the feasibility of $\mathcal{L}$ in time $2^{m \cdot O(m)} =2^{O(m^2)}$ which proves the claim.
 \end{proof}

 \newa{Next, we define a set $S \subseteq \{-1,1\}^m$ to be a
\emph{separable} set if there exists some $w \in \mathbb{S}^{m-1}$
such that $S = \{x \in \{-1,1\}^m : |w \cdot x| >1\}$.
 The next claim says that we can enumerate over (a superset of) the set of all separable sets in time $2^{\tilde{O}(m^2)}$.
 \begin{claim}\label{clm:separable}
 There is an algorithm which runs in time $2^{\tilde{O}(m^2)}$ and lists (a superset of ) all the separable sets of $\{-1,1\}^m$.
  \end{claim}
  \begin{proof}
Consider any separable set $S \subseteq \{-1,1\}^m$, so
there is some $w \in \mathbb{S}^{m-1}$ such that
$S = \{x \in \{-1,1\}^m : |w \cdot x| >1\}$.
Define $S_{+,w} = \{x \in \{-1,1\}^m : w \cdot x >1\}$.
If we define $S'_{+,w}$ to be the set obtained by negating every element
of $S_{+,w}$, then it is easy to see that $S = S_{+,w} \cup S'_{+,w}$.
Thus, it suffices to enumerate over the sets $S_{+,w}$ for all choices
of $w \in \mathbb{S}^{m-1}$, and output $S = S_{+,w} \cup S'_{+,w}$.

Next, we show how to enumerate all such sets $S_{+,w}$.  For this,
given any $w \in \mathbb{S}^{m-1}$, we define $\alpha_w =
\inf_{x \in S_{+,w}} (w \cdot x-1)$. It is easy to see that if we
define $h_w(x) = \sign(w \cdot x -1 -\alpha/2)$, then $S_{+,w} = h_w^{-1}(1)$.
Thus, if we enumerate all possible halfspaces $h$ over $\{-1,1\}^m$ and
list $h^{-1}(1)$ for each halfspace $h$,
then all possible subsets of the form   $h_w^{-1}(1)$
are included in this list. However, it is well known that
there are $2^{O(m^2)}$ halfspaces $h$ over $\{-1,1\}^m$,
and that these halfspaces can be enumerated in time
$2^{O(m^2 \log m)}$ (see e.g. \cite{MORS:10}).
This finishes the proof of the claim.
  \end{proof}
 }

Finally, our algorithm is simply the following:

\medskip


\begin{itemize}

\item \newa{Run the algorithm in Claim~\ref{clm:separable} and arrange the sets in the output in decreasing order of their size.}
\newa{For each set $S$ in this list},

\begin{itemize}

\item Use the algorithm in Corollary~\ref{corr:first-order} to decide
whether there is a vector $w \in \mathbb{S}^{m-1}$ such that
$|w \cdot x| >1$ for all $x \in S$.
\item If there is such a vector $w$ then exit and return $1 - |S|/2^{m-1}$,
else go to the next set $S$.
\end{itemize}

\end{itemize}


\medskip

The total running time of the first step of the algorithm is clearly $2^{\widetilde{O}(m^2)}$. Since the total number of sets in the list is $2^{\tilde{O}(m^2)}$ and every step in the algorithm takes time $2^{O(m^2)}$, hence the total running time is $2^{\widetilde{O}(m^2)}$.

To establish the 
correctness of the algorithm, fix a vector $w^* \in \mathbb{S}^{m-1}$ such that $\Pr_{x \in \{-1,1\}^m}[|w^* \cdot x| \leq 1] = \T(\mathbb{S}^{m-1})$,
i.e. $\Pr_{x \in \{-1,1\}^m}[|w^* \cdot x| \leq 1] \leq
\Pr_{x \in \{-1,1\}^m}[|w \cdot x| \leq 1]$
for all $w \in \mathbb{S}^{m-1}$.  \ignore{By symmetry of the random
variable $v \cdot x$ for every vector $v$,} We have that
 $\Pr_{x \in \{-1,1\}^m}[|w^* \cdot x| > 1] \geq
\Pr_{x \in \{-1,1\}^m}[|w \cdot x| > 1].$
Since our algorithm enumerates over all
sets $S \subseteq \{-1,1\}^m$ in the output of Claim~\ref{clm:separable} in decreasing order of their size, its correctness follows.
 \end{proof}

\section*{Acknowledgements}
We thank Mihalis Yannakakis for helpful conversations.

\bibliography{allrefs}

\newcommand{\etalchar}[1]{$^{#1}$}
\begin{thebibliography}{MORS10}

\bibitem[AHS12]{AHS:12}
N.~Alon, H.~Huang, and B.~Sudakov.
\newblock {Nonnegative k-sums, fractional covers, and probability of small
  deviations}.
\newblock {\em Journal of Combinatorial Theory, Series B}, 102(3):784--796,
  2012.

\bibitem[BD12]{BK:12}
V.~Bentkus and D.~Dzindzalieta.
\newblock {A tight Gaussian bound for weighted sums of Rademacher random
  variables}.
\newblock Technical report, preprint, 2012.

\bibitem[BDD98]{BenDavidDichterman:98}
S.~Ben-David and E.~Dichterman.
\newblock Learning with restricted focus of attention.
\newblock {\em Journal of Computer and System Sciences}, 56(3):277--298, 1998.

\bibitem[Ben04]{Bentkus:04}
V.~Bentkus.
\newblock {On Hoeffding's Inequality}.
\newblock {\em {Annals of Probability}}, 32:1650--1673, 2004.

\bibitem[BKS99]{BKS:99}
I.~Benjamini, G.~Kalai, and O.~Schramm.
\newblock Noise sensitivity of {B}oolean functions and applications to
  percolation.
\newblock {\em Inst. Hautes \'{E}tudes Sci. Publ. Math.}, 90:5--43, 1999.

\bibitem[BTNR02]{BNR02}
A.~Ben-Tal, A.~Nemirovski, and C.~Roos.
\newblock {Robust Solutions of Uncertain Quadratic and Conic-Quadratic
  Problems}.
\newblock {\em {SIAM Journal on Optimization}}, 13(2):535--560, 2002.

\bibitem[DDFS12]{DDFS:12stoc}
A.~De, I.~Diakonikolas, V.~Feldman, and R.~Servedio.
\newblock {Near-optimal solutions for the Chow Parameters Problem and
  low-weight approximation of halfspaces}.
\newblock In {\em Proc. 44th ACM Symposium on Theory of Computing (STOC)},
  pages 729--746, 2012.

\bibitem[DGJ{\etalchar{+}}10]{DGJ+:10}
I.~Diakonikolas, P.~Gopalan, R.~Jaiswal, R.~Servedio, and E.~Viola.
\newblock Bounded independence fools halfspaces.
\newblock {\em SIAM J. on Comput.}, 39(8):3441--3462, 2010.

\bibitem[DS09]{DiakonikolasServedio:09}
I.~Diakonikolas and R.~Servedio.
\newblock Improved approximation of linear threshold functions.
\newblock In {\em Proc.\ 24th Annual IEEE Conference on Computational
  Complexity (CCC)}, pages 161--172, 2009.

\bibitem[Fel68]{Feller}
W.~Feller.
\newblock {\em An introduction to probability theory and its applications}.
\newblock John Wiley \& Sons, 1968.

\bibitem[Fil12]{Filmus-khintchine}
Y.~Filmus.
\newblock {Khintchine-Kahane using Fourier Analysis}.
\newblock Posted at http://www.cs.toronto.edu/\~{ }yuvalf/KK.pdf, 2012.

\bibitem[Gar07]{Garling}
D.J.H. Garling.
\newblock {\em Inequalities: A journey into linear analysis}.
\newblock Cambridge, 2007.

\bibitem[GL94]{GL:94}
C.~Gotsman and N.~Linial.
\newblock Spectral properties of threshold functions.
\newblock {\em Combinatorica}, 14(1):35--50, 1994.

\bibitem[Guy86]{G86}
R.~K. Guy.
\newblock {Any answers anent these analytical enigmas?}
\newblock {\em {American Math Monthly}}, 93:279--281, 1986.

\bibitem[Haa82]{Haagerup82}
U.~Haagerup.
\newblock {The best constants in the Khintchine inequality}.
\newblock {\em {Studia Math.}}, 70:231--283, 1982.

\bibitem[HK92]{HK92}
R.~Holzman and D.~J. Kleitman.
\newblock {On the product of sign vectors and unit vectors}.
\newblock {\em Combinatorica}, 12(3):303--316, 1992.

\bibitem[HK94]{HK:93}
P.~Hitczenko and S.~Kwapie\'n.
\newblock {On the Rademacher series}.
\newblock {\em {Probability in Banach spaces}}, pages 31--36, 1994.
\newblock Birkh{\"a}user, Boston, MA.

\bibitem[HLNZ08]{HLNZ:07}
S.~He, Z.~Luo, J.~Nie, and S.~Zhang.
\newblock {Semidefinite Relaxation Bounds for Indefinite Homogenous Quadratic
  Optimization}.
\newblock {\em {SIAM Journal on Optimization}}, 19:503--523, 2008.

\bibitem[Jac06]{Jackson:06}
J.~Jackson.
\newblock Uniform-distribution learnability of noisy linear threshold functions
  with restricted focus of attention.
\newblock In {\em Proceedings of the Nineteenth Annual Conference on
  Computational Learning Theory (COLT)}, pages 304--318, 2006.

\bibitem[KLO96]{KLO:96}
S.~Kwapie\'{n}, R.~Latala, and K.~Oleszkiewicz.
\newblock {Comparison of moments of sums of independent random variables and
  differential inequalities}.
\newblock {\em J. Funct. Anal.}, 136:258--268, 1996.

\bibitem[KSTJ99]{KST99}
H.~K\"{o}nig, C.~Sch\"{u}tt, and N.~Tomczak-Jaegermann.
\newblock {Projection constants of symmetric spaces and variants of
  Khintchine's inequality}.
\newblock {\em J. Reine Agnew. Math.}, 511:1--42, 1999.

\bibitem[LO94]{LO94}
R.~Latala and K.~Oleszkiewicz.
\newblock {On the best constant in the Khinchin-Kahane inequality}.
\newblock {\em {Studia Math.}}, 109:101--104, 1994.

\bibitem[MORS10]{MORS:10}
K.~Matulef, R.~O'Donnell, R.~Rubinfeld, and R.~Servedio.
\newblock Testing halfspaces.
\newblock {\em SIAM J. on Comput.}, 39(5):2004--2047, 2010.

\bibitem[MS90]{Montgomery-Smith:90}
S.~Montgomery-Smith.
\newblock {The distribution of Rademacher sums}.
\newblock In {\em Proceedings of the American Mathematical Society}, pages
  517--522, 1990.

\bibitem[MT94]{MaassTuran:94}
W.~Maass and G.~Turan.
\newblock How fast can a threshold gate learn?
\newblock In {\em Computational Learning Theory and Natural Learning Systems:
  Volume I: Constraints and Prospects}, pages 381--414. MIT Press, 1994.

\bibitem[O'D12]{ODonnellbook}
R.~O'Donnell.
\newblock Analysis of boolean functions.
\newblock Technical report, 2012.
\newblock Book serialization available at
  http://analysisofbooleanfunctions.org/.

\bibitem[Ole96]{Olek:96}
K.~Oleszkiewicz.
\newblock {On the Stein property of Rademacher sequences}.
\newblock {\em {Probability and Mathematical Statistics}}, 16:127--130, 1996.

\bibitem[Ole99]{Oleszkiewicz:99}
K.~Oleszkiewicz.
\newblock {Comparison of moments via Poincar\'{e}-type inequality}.
\newblock {\em Contemporary Mathematics}, 234:135--148, 1999.

\bibitem[OS11]{OS11:chow}
R.~O'Donnell and R.~Servedio.
\newblock {The Chow Parameters Problem}.
\newblock {\em SIAM J. on Comput.}, 40(1):165--199, 2011.

\bibitem[Per04]{Peres:04}
Y.~Peres.
\newblock Noise stability of weighted majority, 2004.
\newblock Available at http://arxiv.org/abs/math/0412377.

\bibitem[Pin94]{Pinelis:94}
I.~Pinelis.
\newblock Extremal probabilistic problems and {H}otelling's $t^2$ test under a
  symmetry condition.
\newblock {\em Ann. Statist.}, 22:357--368, 1994.

\bibitem[Pin12]{Pinelis:12}
I.~Pinelis.
\newblock {An asymptotically Gaussian bound on the Rademacher tails}.
\newblock {\em {Electronic Journal of Probability.}}, 17:1--22, 2012.

\bibitem[Ren88]{Ren:88}
J.~Renegar.
\newblock {A faster PSPACE algorithm for deciding the existential theory of the
  reals}.
\newblock {\em {IEEE Annual Symposium on Foundations of Computer Science}},
  pages 291--295, 1988.

\bibitem[Ser07]{Servedio:07cc}
R.~Servedio.
\newblock {Every linear threshold function has a low-weight approximator}.
\newblock {\em Comput. Complexity}, 16(2):180--209, 2007.

\bibitem[So09]{So09}
A.~M.-C. So.
\newblock Improved approximation bound for quadratic optimization problems with
  orthogonality constraints.
\newblock In {\em SODA}, pages 1201--1209, 2009.

\bibitem[Sza76]{Szarek:76}
S.J. Szarek.
\newblock {On the best constants in the Khinchine inequality}.
\newblock {\em Studia Math.}, 58:197--208, 1976.

\bibitem[Tom87]{Tomaszewski-Khintchine}
B.~Tomaszewski.
\newblock {A simple and elementary proof of the Khintchine inequality with the
  best constant}.
\newblock {\em {Bull. Sci. Math.}}, 111(2):103--109, 1987.

\end{thebibliography}
\bibliographystyle{alpha}

\end{document}